\documentclass[ejs]{imsart-ps/imsart}

\pubyear{20xx}
\volume{0}
\issue{0}
\firstpage{0}
\lastpage{0}

\RequirePackage[OT1]{fontenc}
\RequirePackage{amsthm,amsmath}
\RequirePackage[numbers]{natbib}
\usepackage{enumerate}
\usepackage{float}
\usepackage{placeins}
\usepackage{macros}
\RequirePackage[colorlinks,citecolor=blue,urlcolor=blue]{hyperref}
\usepackage{xr}
\usepackage{array}
\usepackage{booktabs}
\usepackage{framed}
\usepackage{hyperref}
\usepackage{mathtools}

\graphicspath{ {./figs/} }
\DeclareGraphicsExtensions{.png,.pdf}

\startlocaldefs
\numberwithin{equation}{section}
\theoremstyle{plain}

\endlocaldefs

\begin{document}

\begin{frontmatter}
\title{Profile Likelihood Biclustering}
\runtitle{Profile Likelihood Biclustering}

\begin{aug}
\author{\fnms{Cheryl} \snm{Flynn}\ead[label=e1]{cflynn@research.att.com}}
\address{AT\&T Labs Research\\
New York, NY\\
\printead{e1}}

\author{\fnms{Patrick} \snm{Perry}
\ead[label=e3]{patperry@gmail.com}}

\address{Oscar Health\\
New York, NY\\
\printead{e3}}

\runauthor{C. Flynn and P. Perry}

\affiliation{AT\&T Labs Research and Oscar Health}

\end{aug}



\newcommand{\Keywords}[1]{\par\noindent
{\small{KEY WORDS}: #1}}

\begin{abstract}

Biclustering, the process of simultaneously clustering the rows and columns of
a data matrix, is a popular and effective tool for finding structure in a
high-dimensional dataset.  Many biclustering procedures appear to work well in
practice, but most do not have associated consistency guarantees.  To address
this shortcoming, we propose a new biclustering procedure based on profile
likelihood.  The procedure applies to a broad range of data modalities,
including binary, count, and continuous observations.  We prove that the
procedure recovers the true row and column classes when the dimensions of the
data matrix tend to infinity, even if the functional form of the data
distribution is misspecified.  The procedure requires computing a
combinatorial search, which can be expensive in practice.  Rather than
performing this search directly, we propose a new heuristic optimization
procedure based on the Kernighan-Lin heuristic, which has nice computational
properties and performs well in simulations.  We demonstrate our procedure
with applications to congressional voting records, and microarray analysis.


\end{abstract}


\end{frontmatter}

\section{Introduction}\label{S:introduction}

Suppose we are given a data matrix $\mX = [X_{ij} ]$, and our goal is to
cluster the rows and columns of $\mX$ into meaningful groups.  For example,
$X_{ij}$ can indicate whether or not user~$i$ has an interest in product~$j$,
and our goal is to segment users and products into relevant subgroups.
Alternatively, $X_{ij}$ could be the log activation level of gene~$j$ in
patient~$i$; our goal is to seek groups of patients with similar genetic
profiles, while at the same time finding groups of genes with similar
activation levels.  The general simultaneous clustering problem is known by
many names, including direct clustering~\citep{hartigan72}, block
modeling~\citep{arabie78}, biclustering~\citep{mirkin96}, and
co-clustering~\citep{dhillon01}.

Empirical results from a broad range of disciplines indicate that
biclustering is useful in practice.  \citet{ungar98} and \citet{hofmann99}
found that biclustering helps identify structure in collaborative filtering
contexts with heterogeneous users and sparsely-observed preferences.
\citet{eisen98} used microarray data to simultaneously cluster genes and
conditions, finding that genes with similar functions often cluster together.
\citet{harpaz10} applied biclustering methods to a Food and Drug
Adminstration report database, identifying associations between certain active
ingredients and adverse medical reactions.  Several other applications of
biclustering exist \citep{cheng00,getz00,lazzeroni02,kluger03};
\citet{madeira04} give a comprehensive survey.

Practitioners interested in biclustering have used a variety of different
algorithms to achieve their results.  Clearly, many of these algorithms work
well in practice, but they are often ad-hoc, and there are no rigorous
guarantees as to their performance.  Without these guarantees practicioners cannot be assured
that their discoveries from biclustering will generalize or be reproducible; collecting more data
may lead to completely different clusters.

There are two approaches to evaluating the theoretical performance of these procedures.
The first is to define a higher-level learning task, and evaluate procedures using
a task-dependent measure of generalization performance \citep{seldin09,seldin10}.
We instead consider the alternative approach, which is to consider the problem purely
as an unsupervised learning task.  In this case, the procedure is evaluated based on the identified biclusters, where
a reasonable goal is consistent biclustering.

Our first contribution is to formalize biclustering as an estimation problem.
We do this by introducing a probabilistic model for the data matrix $\mX$,
where up to permutations of the rows and columns, the expectation of $\mX$ has
a block structure.  Next, we make distributional assumptions on the elements
of $\mX$ and derive a clustering objective function via
profile likelihood~\citep{murphy00}.  Finally, we show that the maximum
profile likelihood estimator performs well when the distributional assumptions
are not satisfied, in the sense that it is still consistent.  To our
knowledge, this is the first general consistency result for a biclustering
algorithm.

Unfortunately, it is computationally intractable to compute the maximum
profile likelihood estimator.  It is for this reason that \citet{ungar98}, who
used a similar probabilistic model for the data matrix, dismissed
likelihood-based approaches as computationally infeasible.  Our second
contribution, then, is to propose a new approximation algorithm for finding a
local maximizer of the biclustering profile likelihood.  Our algorithm is
based on the Kernigham-Lin heuristic \citep{kernighan70}, which was employed
by \citet{newman06} for clustering network data. This is a greedy optimization
procedure, so it is not guaranteed to find a global optimum. To mitigate this,
we run the fitting procedure repeatedly with many random initializations; as
we increase the number of initializations, the probability that the algorithm
finds the global optimum increases. We show that this procedure has low
computational complexity and it performs well in practice.


Our work was inspired by clustering methods for
symmetric binary networks.  In that context, $\mX$ is an $n$-by-$n$ symmetric
binary matrix, and the clusters for the rows of $\mX$ are the same as the
clusters for the columns of $\mX$.  \citet{bickel09} used methods similar to
those used for proving consistency of M-estimators to derive results for
network clustering when $n$ tends to infinity.  This work was later extended
by \citet{choi11}, who allow the number of clusters to increase with $n$;
\citet{zhao11a}, who allow for nodes not belonging to any cluster; and
\citet{zhao12}, who incorporate individual-specific effects.  In parallel to this work, 
theoretical and computational advancements have been made for community 
detection in symmetric networks 
using pseudo-likelihood methods \citep{amini2013}, variational methods 
\citep{daudin2008, celisse2012, bickel2013}, belief propagation \citep{mossel2016}, 
spectral clustering \citep{rohe11, fishkind2012, lei2015, jin2015}, and semidefinite 
programming \citep{amini2018, guedon2014}.  We refer the reader to \citet{zhao2018} 
and \citet{abbe2018} for recent surveys.
In the context of biclustering, \citet{rohe12} study spectral methods for
unsymmetric binary networks; \citet{ames14} study recovering the correct partitioning 
of bipartite graphs into disjoint subgraphs; \citet{gao16} study problems related to 
biclustering, but focus on the recovery of the mean matrix rather than recovery of 
the correct biclustering; \citet{choi14} study community detection in binary arrays 
when the blockmodel assumption is misspecified;
\citet{mariadassou15} study latent and stochastic block models where the row 
and column cluster assignments are treated as random variables; and \citet{razaee19} 
study the problem of matched community detection in bipartite networks with node covariates.

In our report we have extended methods originally developed for an extremely
specialized context (symmetric binary networks) to handle clustering for
arbitrary data matrices, including non-symmetic networks and real-valued
data entries.  Using standard conditions, we have been able to
generalize the \citet{bickel09} results beyond Bernoulli random variables.  To
our knowledge, this is the first time methodologies for binary networks have
been used to study general biclustering methods.  Notably, our extensions can
handle a variety of data distributions, and they can handle both dense and
sparse data matrices.

The main text of the paper is organized as follows.  First
Section~\ref{S:set-up} describes the theoretical setup and
Section~\ref{S:consistency-results} presents our main result with a heuristic
proof.  Then, Section~\ref{S:theoretical} describes the formal theoretical
framework and states the rigorous consistency results.  Next,
Section~\ref{S:alg} presents our approximation algorithm.  Using this
algorithm, Section~\ref{S:simulations} corroborates the theoretical findings
through a simulation study, and Section \ref{S:applications} presents
applications to a microarray and a congressional voting dataset.
Section~\ref{S:conclusion} presents some concluding remarks.  The
appendices include additional proofs, empirical results, 
and an application to a movie review dataset.

\section{Estimation problem and criterion functions}\label{S:set-up}

Our first task is to formalize biclustering as an estimation problem.  To
this end, let $\mX = [ X_{ij} ] \in \reals^{m \times n}$ be a data matrix.  We
follow the network clustering literature and posit existence of $K$~\emph{row
classes} and $L$~\emph{column classes}, such that the mean value of entry
$X_{ij}$ is determined solely by the classes of row~$i$ and column~$j$.  That
is, there is an unknown row class membership vector $\vc~\in~K^m$, an unknown
column class membership vector $\vd~\in~L^n$, and an unknown mean matrix
$\mM~=~[\mu_{kl}]~\in~\reals^{K \times L}$ such that
\begin{equation}\label{E:block-model}
  \E X_{ij} = \mu_{c_i d_j}.
\end{equation}
We refer to model~\eqref{E:block-model} as a \emph{block model}, after the
related model for undirected networks proposed by~\citet{holland83}.  Under the
assumptions of the block model, biclustering the rows and columns of the data matrix
is equivalent to estimating $\vc$ and $\vd$.

Not all block models give rise to well-defined estimation problems.  To ensure
that $K$ and $L$ are well-defined, we require that each class has at least one
member, and that no two classes have the same mean vector.  Formally, define
row class proportion vector $\vp \in \reals^{K}$ with element $p_a = m^{-1}
\sum_i \I(c_i = a)$ equal the proportion of nodes with row class $a$.  Also,
define column class proportion vector $\vq \in \reals^{L}$ with element
$q_b~=~n^{-1}\sum_{j}\I(d_j=b)$ equal to the proportion of nodes with column
class $b$.  We require that every element of~$\vp$ and~$\vq$ be nonzero.  To
ensure that the mean vectors of the row classes are distinct, we require that
no two rows of $\mM$ are identical.  Similarly, we require that no two columns
of $\mM$ are identical.

We estimate the clusters by assigning labels to the rows and columns of $\mX$,
codified in vectors $\vg \in K^m$ and $\vh \in L^n$.  Ideally, $\vg$ and $\vh$
match $\vc$ and $\vd$.  Note we are assuming that the true numbers of row and
column clusters, $K$ and $L$, are known, or they have been correctly estimated
by some model selection procedure.  We measure the performance of a particular
label assignment through the corresponding confusion matrix.  Specifically,
for row and column label assignments $\vg$ and $\vh$, define normalized
confusion matrices $\mC \in \R^{K \times K}$ and $\mD \in \R^{L \times L}$
by
\[
  C_{ak}
  =
  \frac{1}{m} \sum_{i} I(c_i = a, g_i = k),
  \qquad
  D_{bl}
  =
  \frac{1}{n} \sum_{j} I(d_j = b, h_j = l).
\]
Entry $C_{ak}$ is the proportion of nodes with class $a$ and label $k$; entry
$D_{bl}$ is defined similarly. These matrices are normalized so that $\mC
\vone = \vp$ and $\mD \vone = \vq$ are the class proportion vectors, and
$\mC^T \vone = \hat{\vp}$ and $\mD^T \vone = \hat{\vq}$ are the label
proportion vectors.  If $\mC$ and $\mD$ are diagonal, then the assigned labels
match the true classes.  More generally, if $\mC$ and $\mD$ can be made diagonal
by permuting their columns, then the partition induced by the labels matches
the partition induced by the classes.  The goal, then, is to find row and
column labellings such that $\mC$ and $\mD$ are permutations of diagonal
matrices.

In practice, we cannot estimate $\mC$ and $\mD$ directly, because we do not
have knowledge of the true row and column classes.  To evaluate the quality of
a biclustering, we need a surrogate criterion function.  Analogously to
\citet{bickel09}, we employ profile likelihood for this purpose.

In Bickel and Chen's setting, the data are binary, so there is a
natural data likelihood which arises from the Bernoulli distribution.  Our
setting is more general, with $\mX_{ij}$ allowed to be a count or a continuous
measurement, so there are many possible choices for the element densities.  We
proceed by initially assuming that the elements of $\mX$ are sampled from
distributions in a single-parameter exponential family.  Conditional on $\vc$
and $\vd$, the elements of $\mX$ are independent, and entry $X_{ij}$ has
density $g(x ; \eta_{c_i d_j})$ with respect to some base measure~$\nu$, where
\[
  g(x; \eta) = \exp\{ x \eta - \psi(\eta)\};
\]
$\psi(\eta)$ is the cumulant generating function, and
$\eta_{kl} = (\psi')^{-1}(\mu_{kl})$ is the natural parameter.  Later, we will
relax the assumption of the specific distributional form.

With labels $\vg$ and $\vh$, the complete data log-likelihood is
\begin{align*}
  \lik(\vg, \vh, \mM) & =
    \sum_{k,l}
        \sum_{i,j}
            \{X_{ij} \eta_{kl} - \psi(\eta_{kl})\}\I(g_i=k,h_j=l)\\
    & = m n
    \sum_{k,l}
      \hat p_k \, \hat q_l \,
      \{ \bar X_{kl} \, \eta_{kl} - \psi(\eta_{kl}) \},
\end{align*}
where
\(
  \hat p_k =
  m^{-1} \sum_i \I(g_i = k)
\)
and
\(
  \hat q_l =
  n^{-1} \sum_j \I(h_j = l)
\)
are the estimated class proportions and
\(
  \bar X_{kl} =
  \{ \sum_{i,j} \I(g_i = k, h_j = l) \}^{-1}
  \sum_{i,j} X_{ij} \I(g_i = k, h_j = l)
\)
is the estimated cluster mean.  We get the profile log-likelihood by
maximizing the log-likelihood over the mean parameter matrix $\mM$:
\[
  \plik(\vg, \vh)
    = \sup_{\mM} \lik(\vg, \vh, \mM)
    = m n \sum_{k,l} \hat p_k \, \hat q_l \, \psi^\ast(\bar X_{kl}),
\]
where $\psi^\ast(x) = \sup_{\eta} \{ x \eta - \psi(\eta) \}$ is the convex
conjugate of $\psi$.  We refer to $\psi^\ast$ as the relative entropy function
since $\psi^\ast(\mu)$ is equal to the Kullback-Leiber divergence of the base
measure~$\nu$ from the distribution in the exponential family with mean~$\mu$
\citep{brown1986}.

Following the above derivation, a natural criterion for the quality of
labeling $(\vg, \vh)$ is the profile log-likelihood $\plik(\vg, \vh)$.  In the
sequel, we consider a far more general setting.  We consider criterion
functions of the form
\begin{equation}\label{E:criterion-general}
  F(\vg, \vh) = \sum_{k,l} \hat p_k \, \hat q_l \, f(\bar X_{kl}),
\end{equation}
where $f$ is any smooth convex function.  Following the derivation above, we
refer to $F$ as a profile likelihood and we refer to $f$ as the corresponding
relative entropy function.  However, we do not assume that likelihood has been
correctly specified.  In particular,  the elements of $\mX$ can have
arbitrary distributional forms under the assumptions of the 
block model~\eqref{E:block-model}, not necessarily belonging to any exponential
family.  We explicitly allow for heteroscedasticity and distributional
misspecification.  We show that under mild technical conditions, the maximizer
of $F$ is a consistent estimator of the true row and column classes.

\section{Heuristic justification}\label{S:consistency-results}

In Section~\ref{S:set-up}, we defined a formal biclustering estimation problem
and we motivated a class of criterion functions for this problem based on
profile likelihood.  In this section, we investigate the behavior of the
criterion functions.  In particular, we outline a heuristic argument which
shows that the row and column labels found by maximizing these criterion
functions are good estimates of the true row and column classes.
Formal statements of the results and their proofs are given in Section~\ref{S:theoretical}
and Appendix~\ref{app-theory}.

As noted in Section~\ref{S:introduction}, the main
thrust of our theoretical results are similar to that used in the literature
on clustering for symmetric binary networks initiated by \citet{bickel09} and
extended by \citet{choi11}, \citet{zhao11a} and \citet{zhao12}.
The main point of departure from this previous work are that we work with
arbitrary data modalities instead of symmetric binary matrices.


Let $\mX \in \reals^{m \times n}$ be a data matrix drawn from an identifiable
block model~\eqref{E:block-model} with row and column classes $\vc \in K^m$
and $\vd \in L^n$ and mean matrix $\mM \in \reals^{K \times L}$.  Let $\vp$,
and $\vq$ be as defined in Section~\ref{S:set-up}.  For any row and column
labeling $\vg$ and $\vh$, let $\vhp$, $\vhq$, and $\mbX$ be the corresponding
estimates of $\vp$, $\vq$, and $\mM$, and let $\mC$ and $\mD$ be the confusion
matrices.  Let $F$ be a profile likelihood criterion function as
in~\eqref{E:criterion-general} with corresponding relative entropy function
$f$, assumed to be smooth and strictly convex.

We now outline a series of results which show that the maximizers of $F$ are
good estimates of the true row and column classes.

\begin{proposition}\label{P:population-crit}
The criterion function $F$ is uniformly close to a ``population criterion
function'' $G$ which only depends on the confusion matrices.
\end{proposition}

If $n$ and $m$ are large, then for any choice of $\vg$ and $\vh$, the
estimated cluster mean $\bar X_{kl}$ will be close to $E_{kl}$, the average
value of $\E X_{ij}$ over the block defined by labels $k$ and $l$.  This
quantity can be computed in terms of the confusion matrices as
\[
  E_{kl}
  =
  \frac{
     \sum_{i,j} \sum_{a,b} \mu_{ab} \I(c_i = a, g_i = k) \I(d_j = b, h_j = l)
  }{
     \sum_{i,j} \I(g_i = k, h_j = l)
  }
  =
  \frac{
    [\mC^\trans \, \mM \, \mD]_{kl}
  }{
    [\mC^\trans \vone]_k  [ \mD^\trans \vone]_l
  }.
\]
By applying Bernstein's inequality, one can show that
$E_{kl}$ is close to $\bar X_{kl}$ uniformly over all
choices of $\vg$ and $\vh$.  Thus, we get the population criterion
function by replacing $\bar X_{kl}$ with $E_{kl}$.

For each non-negative vector $\vt \in \reals_+^{N}$ define
$\sC_{\vt}$ to be the set of $N \times N$ normalized confusion matrices
with fixed row sums:
\(
  \sC_{\vt}
  =
  \{
    \mW \in \R_{+}^{N \times N} : \mW \vone = \vt
  \}.
\)
The population version of $F$ is a function of the row and column confusion
matrices, $G : \sC_{\vp} \times \sC_{\vq} \to \reals$, with
\begin{align*}
  G(\mC, \mD)
  &= \sum_{k,l} [\mC^T \vone]_k \, [\mD^T \vone]_l \,
     f\Big(\frac{[\mC^T \mM \mD]_{kl}}
                {[\mC^T \vone]_k \, [\mD^T \vone]_l}
     \Big).
\end{align*}
Since $\bar X_{kl}$ is uniformly close to $E_{kl}$, under mild regularity
conditions on $f$, the criterion $F(\vg, \vh)$ is uniformly close to
$G(\mC, \mD)$.  Proposition~\ref{L:pop-crit}
contains a rigorous statement of this result.

\begin{proposition}\label{P:self-consistent}
The population criterion function $G$ is self-consistent.
\end{proposition}

Self-consistency is an important property for any criterion function, which
implies that in the absence of noise, the criterion function will be maximized
at the truth~\citep{tarpey96}.  In our context, self-consistency means that $G$
is maximized when $\mC$ and $\mD$ are permutations of diagonal matrices.

The self-consistency of $G$ follows from the strict convexity of $f$:
\begin{align*}
  G(\mC, \mD)
    & = \sum_{k,l} [\mC^T \vone]_k \, [\mD^T \vone]_l \,
        f\Big(\frac{[\mC^T \mM \mD]_{kl}}
                {[\mC^T \vone]_k \, [\mD^T \vone]_l}
        \Big)\\
    & \leq \sum_{k,l} \sum_{a,b} C_{ak} D_{bl} f(\mu_{ab})\\
    & = \sum_{a,b} p_a \, q_b \, f(\mu_{ab}).
\end{align*}
If $\mM$ has no two identical rows and no two identical columns, then exact
equality holds only when $\mC$ and $\mD$ are permutations of diagonal
matrices.  Thus, $G$ is maximized when the row and column class partitions
match the label partitions.  Proposition~\ref{L:pop-perturb} gives
a refined self-consistency result with a more complete characterization of the
behavior of $G$ near its maxima.

\begin{proposition}\label{P:consistent}
Under enough regularity, the maximizer of the criterion function $F$ is close
to the true row and column class partition.
\end{proposition}

This is a direct consequence of
Propositions~\ref{P:population-crit}~and~\ref{P:self-consistent}.  The
criterion $F$ is uniformly close to the population criterion $G$, and $G$ is
maximized at the true class partitions.  Thus, the maximizer of $F$ is close
to the maximizer of $G$.  Importantly, Proposition~\ref{P:consistent} does not
require any distributional assumptions on the data matrix $\mX$ beyond its
expectation satisfying the block model.  In particular this result can be
applied to binary matrices, count data, and continuous data.
Theorems~\ref{T:consistency}~and~\ref{T:finite-sample} contain precise
statements analogous to Proposition~\ref{P:consistent}.

\section{Rigorous Theoretical Results}\label{S:theoretical}

Here we provide formal statements of the main result from
Section~\ref{S:consistency-results}.  The proofs of these results
are contained in Appendix~\ref{app-theory}

We work in an asymptotic framework, where the dimensions of the data matrix
tend to infinity.  Let $\mX_n \in \R^{m \times n}$ be a sequence of data
matrices indexed by $n$, with $m = m(n)$ and $m(n) \to \infty$ as $n \to
\infty$. We will also suppose that $n/m \to \gamma$ for some finite constant
$\gamma>0$; this assumption is not essential, but it simplifies the assumption
and result statements.

Suppose that
for each $n$ there exists a row class membership vector $\vc_n \in K^m$ and a
column class membership vector $\vd_n \in L^n$. We assume that there exist
vectors $\vp \in \R^K$ and $\vq \in \R^L$ such that $\hat p_k(\vc) \to p_k$
and $\hat q_l(\vd) \to q_l$ as $n \to \infty$ almost surely for all $k$ and
$l$; this assumption is satisfied, for example, if the class labels are
independently drawn from a multinomial distribution.  When there is no
ambiguity, we omit the subscript~$n$.

We define the mean matrix $\mM = [ \mu_{kl} ] \in \reals^{K \times L}$ as in
Section~\ref{S:consistency-results}, but allow it to possibly vary with $n$.
To model sparsity in $\mX$, we allow $\mM$ to tend to $\mathbf{0}$.  To avoid
degeneracy, we suppose that there exists a sequence $\rho$ and a fixed matrix
$\mM_0 \in \reals^{K \times L}$ such that~$\rho^{-1}~\mM~\to~\mM_0$.  Denote
by $\sM_0 \in \R$ the convex hull of the entries of $\mM_0$ defined as
\[
\sM_0 \coloneqq \left\{\sum_{k = 1}^K \sum_{l=1}^L \lambda_{kl}[\mM_0]_{kl} 
:
\sum_{k = 1}^K \sum_{l=1}^L \lambda_{kl} = 1, \lambda_{ij} \in [0, 1]\right\}.
\]
Let $\sM$ be a
neighborhood of $\sM_0$.

To adjust for the sparsity, we redefine the criterion and population criterion
functions as
\begin{align*}
  F(\vg, \vh) &= \sum_{k,l} \hat p_k \hat q_l f(\rho^{-1} \bar X_{kl}), \\
  G(\mC, \mD) &= \sum_{k,l} [\mC^\trans \vone]_k [\mD^\trans \vone]_l
                            f\Big(\frac{[\mC^T \mM_0 \mD]_{kl}}
                                       {[\mC^T \vone]_k \, [\mD^T \vone]_l}\Big).
\end{align*}
We discuss these modifications and the role of $\rho$ in Section~\ref{S:rho}.

We only consider nontrivial partitions; to this end, for $\varepsilon > 0$,
define $\sJ_\varepsilon$, the set of nontrivial labellings as
\[
  \sJ_\varepsilon = \{\vg, \vh : \hat p_k(\vg)>\varepsilon,\hat q_l(\vh)>\varepsilon\}.
\]

\subsection{Assumptions}\label{s:assumptions}

We require the following regularity conditions:
\begin{enumerate}
\renewcommand{\theenumi}{(C\arabic{enumi})}
\renewcommand{\labelenumi}{(C\arabic{enumi})}

\item \label{A:identifiable}
The biclusters are identifiable: no two rows of the $\mM_0$ are equal, and no
two columns of $\mM_0$ are equal.

\item \label{A:rate-convex}
The relative entropy function is locally strictly convex: $f''(\mu)>0$ when
$\mu \in \sM$.

\item \label{A:deriv-bound}
The third derivative of the relative entropy function is locally bounded:
$\abs{f'''(\mu)}$ is bounded when $\mu \in \sM$.

\item \label{A:avg-var}
The average variance of the elements is of the same order as $\rho$:
\begin{equation*}
    \limsup_{n \to \infty}
    \frac{1}{\rho m n}
    \sum_{i,j}
      \E[(X_{ij} - \mu_{c_i d_j})^2
        ] < \infty.
\end{equation*}

\item \label{A:sparseness}
The mean matrix does not converge to zero too quickly:
\[
  \limsup_{n\to \infty} \rho \sqrt{nm} = \infty.
\]

\item \label{A:matrix-lindeberg}
The elements satisfy a Lindeberg condition:
for all $\varepsilon > 0$,
\begin{equation*}
   \lim_{n \to \infty}
    \frac{1}{\rho^2 m n}
    \sum_{i,j}
      \E[(X_{ij} - \mu_{c_i d_j})^2
         \I(\abs{X_{ij} - \mu_{c_i d_j}} > \varepsilon \rho \sqrt{m n})
         ] = 0.
\end{equation*}
\end{enumerate}
Condition~\ref{A:identifiable} is necessary for the biclusters to be
identifiable, while \ref{A:rate-convex} and \ref{A:deriv-bound} are mild
regularity conditions on the entropy function.

Condition~\ref{A:avg-var} is trivially satisfied for dense data and is
satisfied for Binomial and Poisson data so long as $\rho^{-1}\mM \to \mM_0$.
However, this condition cannot handle arbitrary sparsity.  For example, if the
elements of $\mX$ are distributed as Negative Binomial random variables, then
condition~\ref{A:avg-var} requires that the product of the mean and the
dispersion parameter does not tend to infinity.  In other words, for sparse
count data, the counts cannot be too heterogeneous.

Condition~\ref{A:sparseness} places a sparsity restriction on the mean matrix.
\citet{zhao12} used the same assumption to establish weak consistency for
network clustering.  A variant Lyaponuv's condition \citep{varadhan01} implies
\ref{A:matrix-lindeberg}.  That is, if
\[
\lim_{n \to \infty} \frac{1}{( \rho \sqrt{m n})^{2+\delta}}
    \sum_{i,j}
    \E |X_{ij} - \mu_{c_i d_j}| ^{2+\delta}
    = 0
\]
for some $\delta>0$, then \ref{A:matrix-lindeberg} follows.  In particular,
for dense data ($\rho$ bounded away from zero), uniformly bounded $(2+\delta)$
absolute central moments for any $\delta > 0$ is sufficient.  For many types
of sparse data, including Bernoulli or Poisson data with $\rho$ converging to
zero, \ref{A:sparseness} is a sufficient condition for
\ref{A:matrix-lindeberg}.

\begin{theorem}\label{T:consistency}
Fix any $\varepsilon > 0$ with $\varepsilon < \min_a \{ p_a \}$ and
$\varepsilon < \min_b \{ q_b \}$.  Let $(\vhg, \vhh)$ satisfy $F(\vhg, \vhh) =
\max_{\sJ_\varepsilon } F(\vg, \vh)$.  If
conditions~\ref{A:identifiable}--\ref{A:matrix-lindeberg} hold, then all limit
points of $\mC(\vhg)$ and $\mD(\vhh)$ are permutations of diagonal matrices,
i.e.~the proportions of mislabeled rows and columns converge to zero in
probability.
\end{theorem}

Our focus is on the cluster assignments, but, using the methods involved to
prove Theorem~\ref{T:consistency}, it is possible to show that when the
assumptions of Theorem~\ref{T:consistency} are satisfied, then the scaled
estimate of the mean, $\rho^{-1} \bar X_{kl}$ converges in probability to the
population quantity. (This follows from
Theorem~\ref{T:consistency} and Lemma~\ref{L:mean-resid-conv}.)

Under stronger distributional assumptions, we can use the methods of the proof
to establish finite-sample results.  For example, if we assume that the
elements of $\mX$ are Gaussian, then the following result holds.

\begin{theorem}\label{T:finite-sample}
Fix any $\varepsilon>0$.  Let $(\hat\vg, \hat\vh)$ satisfy $F(\hat\vg, \hat\vh) = \max_{\sJ_\varepsilon} F(\vg,\vh)$.
If the elements of $\mX$ are independent Gaussian random variables with
constant variance $\sigma^2$ and conditions~\ref{A:identifiable}--\ref{A:deriv-bound} hold, then for any $0<\delta < \min\Big\{1,\frac{8 c \sigma \max\{K^2,L^2\}}{\tau \varepsilon^2}\Big\}$,
\[
\Pr \Big (
    \big(\mC(\vhg), \mD(\vhh)\big) \notin \sP_\delta \cap \sQ_\delta
    \Big)
    \leq
    2K^{m+1}L^{n+1}
            \exp \Big \{
                - \frac{T_n \tau^2 \varepsilon^4 \delta^2}
                    {256 c^2 \sigma^2 \max\{K^4, L^4\}}
            \Big \},
\]
where $c = \sup |f'(\mu)|$ for $\mu$ in $\sM$,
\[
T_n = \inf_{k,l}\Big\{\sum_i \sum_j \I(g_i=k,h_j=l) | (\vg,\vh) \in \sJ_\varepsilon\Big\}.
\]
\end{theorem}

The proof of this finite-sample result follows the same outline as the
asymptotic result; Appendix~\ref{S:finite-sample} gives
details.

\section{Approximation algorithm}\label{S:alg}

Proposition~\ref{P:consistent} shows that the maximizer of the profile
log-likelihood $F$ will give a good estimate of the true clusters.
Unfortunately, finding this maximizer is an NP-hard problem \citep{tanay02}.
Maximizing $F$ is a combinatorial optimization problem with an
exponentially-large state space.  To get around this, we will settle for
finding a local optimum rather than a global one.  We present an algorithm for
finding a local optimum that, in practice, has good estimation performance.

Our approach is based on the Kernighan-Lin heuristic \citep{kernighan70},
which \citet{newman06} used for a related problem, network community
detection.  After inputting initial partitions for the rows and columns, we
iteratively update the cluster assignments in a greedy manner.  The
algorithm works as follows:

\begin{enumerate}

\item Initialize the row and column labels $\vg$ and $\vh$ arbitrarily, and compute $F$.

\item \label{item:alg-loop} Repeat until no local improvement in the profile likelihood is found:

\begin{enumerate}

\item \label{item:alg-step1-a} For each row $i$, determine which of the $K$
  possible label assignments
  for this row is optimal, keeping all other row and column labels fixed. Do not
  perform this assignment, but record the optimal label and the local
  improvement to $F$ that would result if this assignment were to be made.

\item \label{item:alg-step1-b} For each column $j$, determine which of the $L$
  possible label assignments for this column is optimal, keeping all other row
  and column labels fixed. As in step~\ref{item:alg-step1-a}, do not perform
  this assignment, but record the optimal label and the local improvement to
  $F$ that would result if this assignment were to be made.

\item \label{item:alg-step2} In order of the local improvements recorded in
  steps~\ref{item:alg-step1-a} and~\ref{item:alg-step1-b}, sequentially perform the individual cluster reassignments 
determined in these steps, and record the profile likelihood after each
reassignment. Note that these assignments are no longer locally
optimal since the labels of many of the rows and columns change during this
step.  Thus, the profile likelihood could increase or decrease as we move sequentially
through the assignments.


\item Out of the sequence of cluster assignments considered in step~\ref{item:alg-step2},
  choose the one that has the highest profile likelihood.

\end{enumerate}

\end{enumerate}

\noindent
At each complete iteration, a finite sequence of values is considered and we select
the cluster assignment with the highest profile likelihood.  
This allows for efficient consideration of multiple cluster reassignments in each iteration 
of the algorithm as opposed to re-optimizing after each individual row or column reassignment.
The criterion function
increases at each complete iteration and we stop when no local improvement in the profile
likelihood is found.  Thus, the algorithm will converge to a local optimum.

There is no guarantee that the local optimum found by the algorithm will be
the global optimum of the objective function $F$. To mitigate this deficiency,
we will run the algorithm repeatedly with many different random
initializations for $\vg$ and $\vh$.  Each initialization can give rise to a
different local optimum. We choose the cluster assignment with the highest
value of $F$ among all local optima found by the procedure. As we increase the
number of random initializations, the probability that the global optimum will
be in this set will increase.  We found that 100--250 initializations seem to suffice in the simulations and data examples.  Appendix~\ref{app-empirical} provides additional results on the stability of the algorithm.


The main computational bottleneck is updating the value of $F(\vg, \vh)$ as we
update the labels $\vg$ and $\vh$.  We can do this efficiently by storing
and incrementally updating the cluster proportions $\vhp$ and $\vhq$,
the within-cluster row and column sums
\[
  R_{il} = \sum_{j = 1}^{n} X_{ij} \I(h_j = l)
\qquad\text{ and }\qquad
  C_{kj} = \sum_{i = 1}^{m} X_{ij} \I(g_i = k),
\]
and the block sums
\(
  B_{kl} = \sum_{i=1}^{m} R_{il} \I(g_i = k).
\)
Given the values of these quantities, we can compute the criterion $F(\vg,
\vh)$ with $\Oh(K L)$ operations.

If we reassign the label of row $i$ from $k$ to $k'$, then it is
straightforward to update $\vhp$ with $\Oh(1)$ operations.  The values of the
within-cluster row sums $R_{il}$ remain unchanged.   The new values of the
block sums are
\(
  B_{kl}' = B_{kl} - R_{il}
\)
and
\(
  B_{k'l}' = B_{k'l} + R_{il}
\)
for $l = 1, \dotsc, L$; the other block sums are unchanged.  The expensive
part of the update is recomputing the within-cluster column sums $C_{kj}$ for
row labels $k$ and $k'$ and each column $j$.  These computations require $\Oh(m n)$
operations if $\mX$ is dense, and $\Oh(N)$ operations if $\mX$ is sparse with
at most $N$ nonzero elements and $N \geq \max\{ m, n \}$.

Overall we must perform $\Oh(N + KL)$ operations to
reassign the label of row $i$.  Reassigning the label of column $j$ has the
same computational cost.  Thus, one loop iteration in step~\ref{item:alg-loop}
requires $\Oh\big((m + n) (N + KL)\big)$ operations.  For dense data, one
iteration requires $\Oh\big((m + n)(m n + KL)\big)$ operations.  We do not
have an upper bound on the number of iterations until the algorithm converges,
but in our experiments we found that empirically, 25 to 30 iterations suffice.
These iteration counts may seem small, but in fact each iteration performs $m$
row label assignments and $n$ column label assignments. The convergence here
is not a result of early stopping---we found that after 25--30 iterations, no
possible local improvement was possible.

For comparison, a spectral-based biclustering algorithm requires the top
singular vectors of the data matrix, which can be calculated in roughly $\Oh(m n)$
operations using Lanczos or another indirect method \citep{golub96}.

\section{Empirical evaluation}\label{S:simulations}

Here we evaluate the performance of the profile likelihood based biclustering
algorithm from Section~\ref{S:alg}.  We simulate data from a variety of
regimes, including sparse binary data and dense heavy-tailed
continuous measurements.  In these settings, we employ the following three
relative entropy functions:
\begin{subequations}
\begin{align}
  f_{\text{Bernoulli}}(\mu) &= \mu \log \mu + (1 - \mu) \log (1 - \mu),
  \label{E:f-bernoulli} \\
  f_{\text{Poisson}}(\mu) &= \mu \log \mu - \mu, \label{E:f-poisson} \\
  f_{\text{Gaussian}}(\mu) &= \mu^2 / 2 \label{E:f-gaussian}.
\end{align}
\end{subequations}
We evaluate performance both when the profile likelihood is correctly
specified and when the relative entropy function does not match the data
distribution.

In our simulations, we report the proportion of misclassified rows and columns
by the profile likelihood based method (PL).  We initialize partitions randomly, 
and then run the improvement algorithm from Section~\ref{S:alg} until convergence.
We use multiple random starting values to minimize the possibility of finding
a non-optimal stationary point.  Our code for implementing the profile likelihood based
biclustering method is available on GitHub (https://github.com/patperry/biclustpl).

We compare our method to three other biclustering algorithms.  The first
algorithm is a spectral biclustering algorithm, DI-SIM, motivated by a block
model similar to ours \citep{rohe12}.  The algorithm finds the singular value
decomposition of the data matrix $\mX$, and then applies $k$-means clustering
to the top left and right singular vector loadings.  
The second algorithm, Sparse Biclustering (SBC) \citep{tan2014}, assumes the data follows a normal distribution and maximizes the $\ell_1$-penalized log likelihood.  The $\ell_1$-penalty allows for sparse biclusters, where the amount of sparsity is controlled by a regularization parameter $\lambda \geq 0$.  We implement the method using the R 
\texttt{sparseBC} package.  In our simulations, we evaluate the performance of SBC over a grid of $\lambda$ values and report the best case performance of the method.  
Lastly, we compare against KM \citep{macqueen67}, which ignores the interactions between the clusters and
applies $k$-means separately to the rows and columns of $\mX$.

In our first simulation, we generate sparse Poisson count data from a block
model with $K = 2$ row clusters and $L = 3$ column clusters.  We vary the
number of columns, $n$, between 200 to 1400 and we take the number of rows to
be $m = \gamma  n$ for $\gamma \in \{ 0.5,1,2 \}$.  To assign the true row and
column classes $\vc$ and $\vd$, we sample independent multinomials with
probabilities $\vp = (0.3, 0.7)$ and $\vq = (0.2, 0.3, 0.5)$.  We choose the
matrix of block parameters to be
\[
  \mM = [\mu_{ab} ]
  =
  \frac{b}{\sqrt{n}}
  \begin{pmatrix}
    0.92 & 0.77 & 1.66  \\
    0.17 & 1.41 & 1.45
  \end{pmatrix},
\]
where $b$ is chosen between 5 and 20; the entries of the matrix were chosen
randomly, uniformly on the interval $[0,2]$.   We chose the $1/\sqrt{n}$
scaling so that the data matrix would be sparse, with $\Oh(\sqrt{n})$ elements
in each row.  We generate the data conditional on the row and
column classes as
\(
  X_{ij} \mid \vc, \vd \sim \mathrm{Poisson}(\mu_{c_i d_j}).
\)
We run all three methods with $250$ random starting values.

\begin{figure}
\centering
\includegraphics[scale = 1, trim = 0mm 20mm 0mm 0mm]{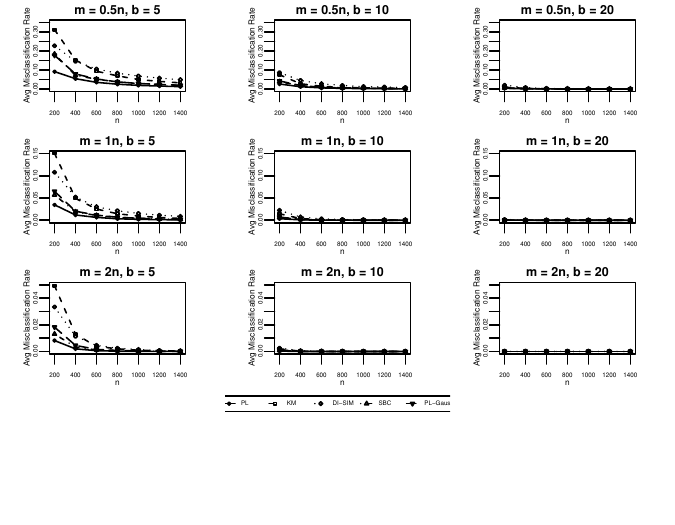}
\caption{Average misclassification rates for Poisson example over 100
simulations.}
\label{poissonSim}
\end{figure}

Figure~\ref{poissonSim} presents the average bicluster misclassification
rates for each sample size and Table~\ref{tabPois1} reports the standard
deviations.  
In all of the scenarios considered, biclustering based on the poisson
profile likelihood criterion performs at least as well as the other methods.
When the relative entropy function is misspecified (using PL-Gaus instead
of PL-Pois), biclustering based on the profile likelihood criterion performs
at least as well as DI-SIM and KM, but SBC has a lower average misclassification
rate in some scenarios when $n$ is small.  As $n$ increases, SBC tends to
select a regularization parameter close to zero, and the two methods perform similarly.
Moreover by looking at the standard deviations, we see that for
the PL methods, the misclassification rate seems to be converging to zero as
we increase $n$.

\begin{table}[H]
\centering
\scalebox{0.6}{
\begin{tabular}{|l|ccc|ccc|ccc|ccc|ccc|}
\hline
\multicolumn{16}{|c|}{$m=0.5n$}\\
\hline
& \multicolumn{3}{c}{PL-Pois} & \multicolumn{3}{|c|}{PL-Norm} &	\multicolumn{3}{|c|}{KM} & \multicolumn{3}{|c|}{DS} & \multicolumn{3}{|c|}{SBC} \\
\hline
n    & b=5  & b=10 & b=20 & b=5  & b=10 & b=20 & b=5   & b=10 & b=20 & b=5   & b=10  & b=20  & b=5   & b=10  & b=20 \\
200	&	0.0356	&	0.0147	&	0.0047	&	0.0878	&	0.0239	&	0.0070	&	0.0957	&	0.0438	&	0.0097	&	0.0403	&	0.0286	&	0.0114	&	0.0935	&	0.0204	 &	0.0063	\\
400	&	0.0182	&	0.0064	&	0.0013	&	0.0313	&	0.0091	&	0.0025	&	0.0478	&	0.0145	&	0.0033	&	0.0279	&	0.0144	&	0.0043	&	0.0238	&	0.0089	 &	0.0025	\\
600	&	0.0103	&	0.0034	&	0.0004	&	0.0153	&	0.0045	&	0.0011	&	0.0253	&	0.0068	&	0.0013	&	0.0186	&	0.0083	&	0.0026	&	0.0138	&	0.0043	 &	0.0011	\\
800	&	0.0076	&	0.0024	&	0.0003	&	0.0104	&	0.0039	&	0.0006	&	0.0188	&	0.0050	&	0.0009	&	0.0140	&	0.0069	&	0.0013	&	0.0102	&	0.0040	 &	0.0006	\\
1000	&	0.0054	&	0.0018	&	0.0003	&	0.0076	&	0.0029	&	0.0003	&	0.0122	&	0.0033	&	0.0004	&	0.0111	&	0.0040	&	0.0008	&	0.0076	&	0.0029	 &	0.0003	\\
1200	&	0.0041	&	0.0010	&	0.0001	&	0.0061	&	0.0017	&	0.0003	&	0.0092	&	0.0021	&	0.0003	&	0.0092	&	0.0034	&	0.0007	&	0.0062	&	0.0017	 &	0.0003	\\
1400	&	0.0036	&	0.0009	&	0.0001	&	0.0047	&	0.0014	&	0.0001	&	0.0077	&	0.0016	&	0.0001	&	0.0071	&	0.0023	&	0.0005	&	0.0047	&	0.0014	 &	0.0001	\\
\hline
\multicolumn{16}{c}{}\\
\hline
\multicolumn{16}{|c|}{$m=n$}\\
\hline
n    & b=5  & b=10 & b=20 & b=5  & b=10 & b=20 & b=5   & b=10 & b=20 & b=5   & b=10  & b=20  & b=5   & b=10  & b=20 \\
200	&	0.0160	&	0.0039	&	0.0009	&	0.0415	&	0.0069	&	0.0010	&	0.0663	&	0.0119	&	0.0019	&	0.0339	&	0.0128	&	0.0030	&	0.0345	&	0.0061 	&	0.0010	\\
400	&	0.0071	&	0.0017	&	0.0000	&	0.0101	&	0.0021	&	0.0004	&	0.0183	&	0.0034	&	0.0004	&	0.0152	&	0.0049	&	0.0009	&	0.0092	&	0.0021	 &	0.0004	\\
600	&	0.0036	&	0.0005	&	0.0000	&	0.0051	&	0.0008	&	0.0000	&	0.0100	&	0.0014	&	0.0000	&	0.0088	&	0.0024	&	0.0002	&	0.0050	&	0.0008	 &	0.0000	\\
800	&	0.0021	&	0.0003	&	0.0000	&	0.0035	&	0.0005	&	0.0000	&	0.0056	&	0.0009	&	0.0000	&	0.0064	&	0.0014	&	0.0001	&	0.0034	&	0.0005	 &	0.0000	\\
1000	&	0.0017	&	0.0002	&	0.0000	&	0.0026	&	0.0004	&	0.0000	&	0.0043	&	0.0006	&	0.0000	&	0.0050	&	0.0009	&	0.0000	&	0.0026	&	0.0003 	&	0.0000	\\
1200	&	0.0012	&	0.0001	&	0.0000	&	0.0020	&	0.0002	&	0.0000	&	0.0030	&	0.0004	&	0.0000	&	0.0036	&	0.0008	&	0.0000	&	0.0020	&	0.0002	 &	0.0000	\\
1400	&	0.0007	&	0.0001	&	0.0000	&	0.0012	&	0.0001	&	0.0000	&	0.0018	&	0.0002	&	0.0000	&	0.0025	&	0.0005	&	0.0000	&	0.0012	&	0.0001	 &	0.0000	\\
\hline
\multicolumn{13}{c}{}\\
\hline
\multicolumn{13}{|c|}{$m=2n$}\\
\hline
n    & b=5  & b=10 & b=20 & b=5   & b=10 & b=20 & b=5  & b=10 & b=20 & b=5   & b=10 & b=20  & b=5   & b=10  & b=20 \\
200	&	0.0059	&	0.0010	&	0.0000	&	0.0118	&	0.0018	&	0.0000	&	0.0204	&	0.0027	&	0.0000	&	0.0141	&	0.0039	&	0.0000	&	0.0096	&	0.0017	&	0.0000	\\
400	&	0.0022	&	0.0001	&	0.0000	&	0.0033	&	0.0003	&	0.0000	&	0.0071	&	0.0005	&	0.0000	&	0.0056	&	0.0010	&	0.0000	&	0.0028	&	0.0003	&	0.0000	\\
600	&	0.0010	&	0.0000	&	0.0000	&	0.0014	&	0.0001	&	0.0000	&	0.0024	&	0.0002	&	0.0000	&	0.0025	&	0.0004	&	0.0000	&	0.0011	&	0.0001	&	0.0000	\\
800	&	0.0004	&	0.0000	&	0.0000	&	0.0008	&	0.0000	&	0.0000	&	0.0015	&	0.0000	&	0.0000	&	0.0017	&	0.0001	&	0.0000	&	0.0007	&	0.0000	&	0.0000	\\
1000	&	0.0002	&	0.0000	&	0.0000	&	0.0005	&	0.0000	&	0.0000	&	0.0009	&	0.0000	&	0.0000	&	0.0013	&	0.0000	&	0.0000	&	0.0004	&	0.0000	&	0.0000	\\
1200	&	0.0002	&	0.0000	&	0.0000	&	0.0003	&	0.0001	&	0.0000	&	0.0006	&	0.0001	&	0.0000	&	0.0009	&	0.0001	&	0.0000	&	0.0003	&	0.0001	&	0.0000	\\
1400	&	0.0000	&	0.0000	&	0.0000	&	0.0002	&	0.0000	&	0.0000	&	0.0005	&	0.0000	&	0.0000	&	0.0006	&	0.0000	&	0.0000	&	0.0002	&	0.0000	&	0.0000	\\
\hline
\end{tabular}}

\caption{Standard deviations for Poisson example over 100 simulations
}\label{tabPois1}
\end{table}

Appendix~\ref{app-empirical} describes in detail the simulations for Bernoulli,
Gaussian, and heavy-tailed $t$ data.  These
results are similar to the Poisson case.  Our method performs at least as well
as the other procedures in all cases and shows signs of convergence.

%

Overall, the simulations confirm the conclusions of
Proposition~\ref{P:consistent}, and they show that our approximate
maximization algorithm performs well.  These results give us confidence that
profile likelihood based biclustering can be used in practice.


\section{Applications}\label{S:applications}

In this section we use profile-likelihood-based biclustering to reveal
structure in two high-dimensional datasets.  For each example, we maximize the
profile log-likelihood using the algorithm described in
Section~\ref{S:alg}.  An additional application example is provided in Appendix~\ref{app-empirical}.

In all three examples, we compare the biclusters selected using the profile-likelihood 
based procedure to those selected by DI-SIM, SBC and KM.
For SBC, we compare performance using the $K$ and $L$ selected for PL and 
the $K$ and $L$ selected using the function provided in the R package.
We also compare the biclusters selected by PL to those selected by
the Convex Biclustering Algorithm (CBC) \citep{chi2017}.  
CBC is a biclustering procedure designed for continuous data.  The algorithm is based on a convex 
criterion function and computes a biclustering solution
path for varying values of a tuning parameter, $\gamma$.  We implement CBC
using the R package \texttt{cvxbiclustr}.  The selected $K$ and $L$ are not specified explicitly,
and are instead dependent on $\gamma$.  In our experiments, we select $\gamma$ using the 
function provided in the R package.  All other parameters are set to the default package settings.


\subsection{GovTrack}

In our first application of the proposed method, we cluster legislators and
motions based on the roll-call votes from the 113th United States House of
Representatives (years 2013--2014). We validate our method by showing that the
clusters found by the method agree with the political parties of the
legislators.

After downloading the roll-call votes from \url{govtrack.org}, we form a data
matrix $X$ with rows corresponding to the 444 legislators who voted in the
House of Representatives, and columns corresponding to the 545 motions voted
upon. Even though there are only 435 seats in the House of Representatives, 9
legislators were replaced mid-session when they resigned or died.
We code the non-missing votes as
\[
  X_{ij} =
  \begin{cases}
    1 &\text{if legislator $i$ voted ``Yea'' on motion $j$,} \\
    0 &\text{if legislator $i$ voted ``Nay'' on motion $j$,} \\
    \text{NA} & \text{if legislator $i$ did not vote on motion $j$}.
  \end{cases}
\]
Not all legislators vote on all motions, and 7\% of the data matrix entries
are missing, coded as $\text{NA}$.  If we assume that the missing data
mechanism is ignorable, then it is straightforward to extend our fitting
procedure to handle incomplete data matrices. Specifically, to handle the
missing data, we replace sums over all matrix entries with sums over all
non-missing matrix entries.

To choose the number of row clusters, $K$, and the number of column clusters,
$L$, we fit all $100$ candidate models with $K$ and $L$ each ranging between
$1$ and $10$. Figure~\ref{screeGovTrack} shows the deviance (twice the
negative log-likelihood) plotted as a function of $K$ and $L$. The left
``scree'' plot shows that for most values of $K$, increasing $L$ from $1$ to
$4$ has a large effect of the deviance, but increases $L$ from $4$ to a larger
value has only a minor effect. Similarly, the right plot shows that increasing
$K$ from $1$ to $2$ causes a large drop in the deviance, but increasing $K$ to
a larger value has only a minor effect. Together, these plot suggest that we
should pick $K = 2$ and $L = 4$.

\begin{figure}
\centering
 \includegraphics[scale=0.8]{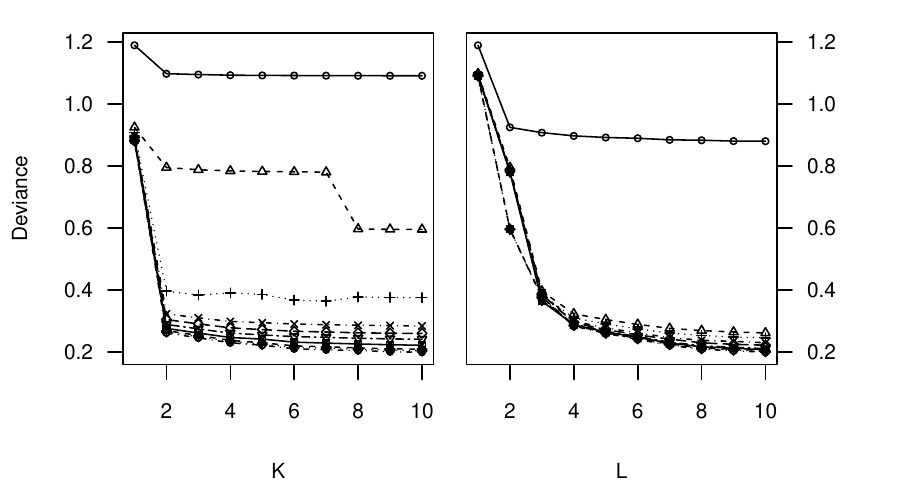}
\caption{GovTrack likelihood for different values of $K$ and $L$}
\label{screeGovTrack}
\end{figure}

%

To guard against the possibility of the bicluster algorithm finding a local
rather than global optimum, we used $100$ random initializations, and chose
the row and column cluster assignments with the highest log-likelihood. To
check the robustness of this assignment, we increased the number of 
random initializations up to $1000$. Even with $10$ times as many restarts, we
still found the same optimal log-likelihood.

Figure~\ref{hmGovTrack} shows a heatmap constructed after biclustering the
data into $2$ row clusters and $4$ column clusters.  The two row-clusters
found by the algorithm completely recover the political parties of the
legislators (every legislator in row cluster $1$ is a Democrat, and every
legislator in row cluster $2$ is a Republican). The fact that we were able to
recover the political party provides us some confidence that the algorithm can
find meaningful clusters in practice. The column clusters reveal four types of
motions:
(1) motions with strong support from both parties;
(2) motions with moderate support from both parties;
(3) motions with strong Democratic support;
(4) motions with strong Republican support.

\begin{figure} 
\centering
\includegraphics[trim = 0mm 15mm 0mm 0mm,scale=0.9]{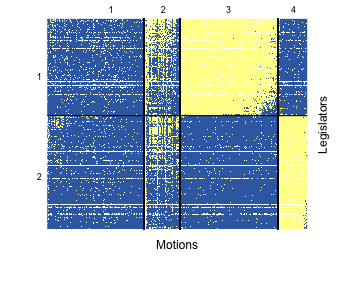}
\caption{Heatmap generated from GovTrack data reflecting the voting patterns
  in the in the different biclusters.
Blue (dark) identifies ``Yea'' votes, yellow (light) identifies ``Nay'' votes,
and white identifies missing votes.}\label{hmGovTrack}
\end{figure}

We compared the clusters found by our method with those found by the four
competing methods. The competing methods do not directly
handle missing data, so for these methods we code ``Yea'' as $+1$, ``Nay`` as
$-1$, and ``Missing'' as $0$. 

Setting $K=2$ and $L=4$, 
DI-SIM returns the same row-clusters, but
classified 50 columns (9\%) differently. KM placed 10 Republicans into the
majority-Democrat cluster; it classified 45 motions (8\%) differently from
profile likelihood.  SBC classified 3 Republicans into the 
majority-Democrat cluster, and had less
agreement with the column-clusters with 137 motions (25.1\%) classified
differently from PL.  

Allowing $K$ and $L$ to vary, SBC and CBC all selected more row 
and column clusters.  SBC selected 9 row clusters, where only one
cluster contained a mix of Republicans and Democrats, and 11 column clusters.
Although the number of clusters differ, the Rand Index between the row clusters
and column clusters found by PL and SBC is 0.746 and 0.653, respectively, suggesting
that similar clusterings are found by the two methods.  On there other hand, CBC 
selected 111 row clusters and 520 column clusters, where many of the
clusters only contained one row/column.  CBC is designed for continuous data
and the default settings do not appear to produce meaningful clusters in this example.  

Overall, the agreement between PL and DI-SIM, KM and SBC gives us confidence that the column clusters are meaningful.

\FloatBarrier

\subsection{AGEMAP}

Biclustering is commonly used for microarray data to visualize the activation
patterns of thousands of genes simultaneously.  It is used to identify patterns
and discover distinguishing properties between genes and individuals.
We use the AGEMAP dataset \citep{zahn07} to illustrate this process.

AGEMAP is a large microarray data set containing the log expression levels for
40 mice across 8,932 genes measured on 16 different tissue types.  For this
analysis, we restrict attention to two tissue types: cerebellum and cerebrum.
The 40 mice in the dataset belong to four age groups, with five males and five
females in each group.  One of the mice is missing data for the cerebrum
tissue so it has been removed from the dataset.

Our goal is to uncover structure in the gene expression matrix. 
We bicluster the 39 $\times$ 17,864 residual matrix
computed from the least squares solution to the multiple linear regression model
\[
  Y_{ij} = \beta_{0j} + \beta_{1j}A_i + \beta_{2j}S_i + \varepsilon_{ij},
\]
where $Y_{ij}$ is the log-activation of gene~$j$ in mouse~$i$,
$A_i$ is the age of mouse~$i$, $S_i$ indicates if mouse~$i$ is male,
$\varepsilon_{ij}$ is a random error, and $(\beta_{0j}, \beta_{1j},
\beta_{2j})$ is a gene-specific coefficient vector. Here, we are biclustering
the residual matrix rather than the raw gene expression matrix because we are
interested in the structure remaining after adjusting for the observed
covariates.

The entries of the residual matrix are not independent (for example, the sum
of each column is zero).  Also, the responses of many genes are likely
correlated with each other.  Thus, the block model assumption required by
Theorem~\ref{T:consistency} is not satisfied, so its conclusion will not hold
unless the dependence between the residuals is negligible.  In light of this
caveat, the example should be considered as exploratory data analysis.

We perform biclustering using profile likelihood based on the Gaussian
criterion~\eqref{E:f-gaussian} with 100 random initializations.
To determine an appropriate number of mice clusters, $K$, and gene clusters, $L$, we experiment
with values of $K$ and $L$ between $1$ and $15$.  Figure~\ref{screeAgeMap} presents the scree plots.
The left plot shows that increasing $L$ beyond $5$ has a relatively small impact on the deviance, and  
similarly, the right plot shows that increasing $K$ beyond $3$ has a relatively minor effect.
This suggests we should set $K = 3$ and $L = 5$.  For this choice of $K$ and $L$, we experimented with using up to 1000 random starting values, but found no change to the resulting log-likelihood.

\begin{figure}[H]
  \includegraphics[scale=0.8]{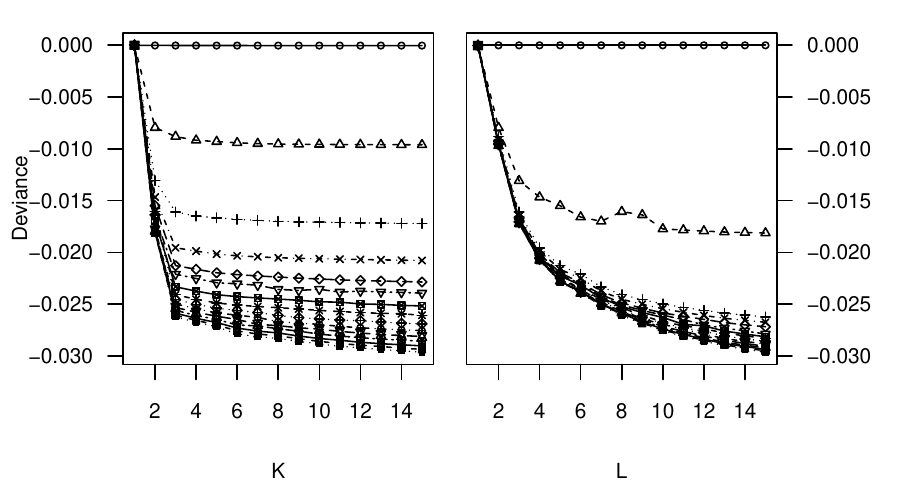}
  \caption{AgeMap likelihood for different values of $K$ and $L$}
  \label{screeAgeMap}
\end{figure}


\begin{figure}[ht]\label{hmAGEMAP}
\centering
\includegraphics[trim = 0mm 10mm 0mm 0mm,scale=.9]{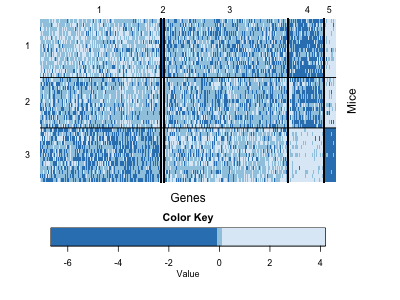}
\caption{Heatmap generated from AGEMAP residual data reflecting the varying
expression patterns in the different biclusters.  The
colors for the matrix entries correspond encode to the first quartile, the middle two quartiles, and the upper quartile.}
\end{figure}

The heatmap presented in Figure~\ref{hmAGEMAP} shows the results.  The expression levels for gene group 1 and 3 appear to be fairly neutral across the three mouse groups, but the other three gene groups have a more visually apparent pattern.  It appears that a mouse can have expression levels for at most two of gene groups 2, 4, and 5.  Mouse group~2 has high expression for gene groups~2~and~4; mouse group~2 has high expression for gene group~4; and mouse group 3 has high expression for gene group~5.

With the exception of CBC, the competing methods found similar gene groups to those found by PL.  
Setting $L = 5$, KM and SBC generally agree with
PL, with 89.5\% and 99.7\% cluster agreement, respectively.  DI-SIM agreed less, but still had 62.5\%
cluster agreement with PL.  Allowing the number of clusters to to vary, SBC still selects a relatively small number of gene groups
with $L = 8$, with a Rand Index with PL of 0.695.  

CBC selected a much larger number of gene groups than any of the competing methods, with $L = 394$. The
gene groups include 1 very large cluster with 16,870 genes and 351 clusters each with only 1 gene.
Based on discussions with the authors, the default selections for CBC may not be appropriate
when the dimensions of the matrix vary by such a large amount.  Adjusting the weights used in the method
could produce more meaningful results, but this goes beyond the scope of our analysis. 


CBC, DI-SIM, KM and SBC all select the same mice groups as PL, even when $K$ is allowed
to vary for CBC and SBC.  The agreement between all of the methods suggests that the three mice groups
are meaningful.  The three clusters of mice found by the methods also agree with those found by \citet{perry10}.  That
analysis identified the mouse clusters, but could not attribute meaning to them.
The bicluster based analysis has deepened our understanding of the three mouse
clusters while suggesting some interesting interactions between the genes.



\section{Discussion}\label{S:conclusion}

We have developed a statistical setting for studying the performance of
biclustering algorithms.  Under the assumption that the data follows a
stochastic block model, we derived sufficient conditions for an algorithm
based on maximizing a profile-likelihood based criterion function to be
consistent.  This is the first theoretical guarantee for any biclustering
algorithm which can be applied to a broad range of data distributions and can
handle both sparse and dense data matrices.

Since maximizing the criterion function exactly is computationally infeasible,
we have proposed an approximate algorithm for obtaining a local optimum rather
than a global one.  We have shown through simulations that the
approximation algorithm has good performance in practice.  Our empirical
comparisons demonstrated that the method performs well in a variety of
situations and can outperform existing procedures.

Applying the profile-likelihood based biclustering algorithm to real data
revealed several interesting findings.  
Our results from the GovTrack dataset 
demonstrated our methods ability to recover ground truth labels when available,
and identified motion clusters that were robust across different methods.
Biclustering the genes and mice in the AGEMAP data exposed
an interesting pattern in the expression of certain genes and we found that at
most two gene groups can have high expression levels for any one mouse.  The
consistency theorem proved in this report gives conditions under which we can
have confidence in the robustness of these findings.

%


\appendix
\section{Additional Proofs and Theoretical Results}\label{app-theory}

\subsection{Proof of Formal Consistency Theorem}

Our proof of the main theoretical results follows the same outline as
Section~\ref{S:consistency-results}.  In particular,
Proposition~\ref{L:pop-crit} is a rigorous statement analogous to
Proposition~\ref{P:population-crit}; Proposition~\ref{L:pop-perturb} is
analogous to Proposition~\ref{P:self-consistent};
Theorems~\ref{T:consistency}~and~\ref{T:finite-sample} are analogous to
Proposition~\ref{P:consistent}.

\subsubsection{Population criterion}

Here, we give a rigorous statement of Proposition~\ref{P:population-crit}.
That is, we establish that in the limit, $F$ is close to its nonrandom population
version, $G$, which depends only on the confusion matrices.

We will need some additional notation and a concentration result.  Define
normalized residual matrix $\mR(\vg, \vh) \in \R^{K \times L}$ by
\[
  \mR(\vg, \vh) = \rho^{-1} \{ \bar \mX(\vg, \vh) - \mE(\vg, \vh) \}.
\]
The law of large numbers establishes that for fixed $\vg$ and $\vh$, the
convergence $R_{kl}(\vg, \vh) \toP 0$ holds.  We can prove a stronger
concentration result, that this convergence is uniform over all $\vg$ and
$\vh$.

\begin{lemma}\label{L:mean-resid-conv}
Under conditions~\ref{A:identifiable}--\ref{A:matrix-lindeberg}, for all
$\varepsilon>0$,
\[
  \sup_{\sJ_\varepsilon }
  \|\mR(\vg, \vh)\|_\infty
  \toP 0,
\]
where $\|\mathbf{A}\|_\infty = \max_{k,l} |A_{kl}|$ for any matrix $\mathbf{A}$.
\end{lemma}

\begin{proof}
For all $t>0$,
\[
  \Pr\Big( \sup_{\sJ_\varepsilon } \norm{ \mR( \vg, \vh ) }_\infty > t \Big)\\
    \leq
        KL \Pr\Big( \sup_{I \in \sI_n}
            \rho^{-1} \Big|\sum_{ \{i,j\} \in I } (X_{ij} - \mu_{c_i d_j})\Big| > t\abs{I} \Big),
\]
where $\sI_n \subset 2^{[n]} \times 2^{[m]}$ is the set of all biclusters
such that $\hat{p}_k > \varepsilon$ for all $k$ and $\hat{q}_l > \varepsilon$
for all $l$. Since $\sI_n$ is a subset of the power set $2^{[nm]}$, by
Lemma~\ref{L:refined-bernstein} in Appendix~\ref{S:technical}, it follows that
the right hand side tends to zero.
\end{proof}

With Lemma~\ref{L:mean-resid-conv}, we can establish that in the limit,
$F(\vg, \vh)$ is close to $G(\mC, \mD)$.

\begin{proposition}\label{L:pop-crit}
  $F$ is close to its population version in the sense that, for all $\varepsilon>0$,
    \[
\sup_{\sJ_\varepsilon }
| F(\vg, \vh)
- G(\mC, \mD)|
\toP 0.
\]
\end{proposition}

\begin{proof}
  The technical assumptions of $f$ imply that its first derivative is bounded.
  Therefore, $f$ is locally Lipschitz continuous with Lipschitz constant $c=\sup |f'(\mu)|$ for $\mu$ in a neighborhood of $\sM$ and
  \[
  \left | F(\vg,\vh)
    - G(\mC, \mD)
    \right |
        \leq c \| \mR(\vg,\vh) \|_\infty.
  \]
  From Proposition \ref{L:mean-resid-conv}, the right hand side converges to zero almost surely and the result follows.
\end{proof}

\subsubsection{Self-consistency}

Once we have established that $F$ is close to its population version, our next
task is to show that the population version is self-consistent.  We will need
a more precise version of Proposition~\ref{P:self-consistent}.

To state the result, for $\delta>0$, define
\begin{subequations}\label{E:confusion-nbhd}
\begin{align}
  \sP_\delta &= \{ \mC \in \sC_{\vp} : \max_{a \neq a'} C_{ak}C_{a'k} < \delta \}, \\
  \sQ_\delta &= \{ \mD \in \sC_{\vq} : \max_{b \neq b'} D_{bl}D_{b'l} < \delta \}.
\end{align}
\end{subequations}
A permutation of a diagonal matrix has only one non-zero entry in each column, so taking $\delta$ close to zero restricts the confusion matrices to be close to permutations of diagonal matrices.

\begin{proposition}\label{L:pop-perturb}
  If
  \(
    \min_a \{ p_a \} > \eta,
  \)
  \(
    \min_b \{ q_b \} > \eta,
  \)
  and $(\mC, \mD) \notin \sP_\delta \times \sQ_\delta$,
  then
   $G(\mC, \mD)$ is small, in the sense that
  \[
     G(\mC, \mD)
     - \sum_{a,b} p_a \, q_b \, f(\rho^{-1} \mu_{ab})
     \leq - \kappa \eta^2 \delta,
  \]
  where $\kappa$ is a constant independent of $\delta$ and $\eta$.
\end{proposition}

\begin{proof}
If $\mD \notin \sQ_\delta$, then for some $l$ and some $b \neq b'$, $D_{bl}
D_{b'l} \geq \delta$.  Since no two columns of $\mM$ are identical, there
exists an $a$ such that $\mu_{ab} \neq \mu_{ab'}$. Let $k$ be the index of the
largest element in row $a$ of matrix $\mC$; this element must be at least as
large as the mean, i.e.\
\[
  C_{ak} \geq \frac{[\mC \vone]_a}{K} \geq \frac{\eta}{K}.
\]
Let $W = [\mC^T \vone]_k [\mD^T \vone]_l$; this is nonzero.  Now, there exists
$\mu_\ast \in \sM$ such that
\[
  [\mC^T \mM \mD]_{kl}
  = C_{ak} D_{bl} \mu_{ab} + C_{ak} D_{b'l} \mu_{ab'}
  + (W - C_{ak} D_{bl} - C_{ak} D_{b'l}) \mu_\ast.
\]
Let $z = [\mC^T \mM \mD]_{kl} / W$.  Set
\(
  \kappa_0 = \inf_{\mu \in \sM} f''(\mu)
\)
and define $\mN = [\nu_{ab}] \in \R^{A \times B}$ with $\nu_{ab} =
f(\rho^{-1} \mu_{ab})$.
By a refined Jensen's inequality (Lemma~\ref{L:refined-jensen} in
Appendix~\ref{S:technical}),
it follows that
\begin{align*}
\frac{[\mC^T \mN \mD]_{kl}}{W} -  f(z)
& \geq \kappa_0\frac{C^2_{ak} D_{bl} D_{b'l}}{W^2} \Big( \frac{1}{2} (\mu_{ab}-z)^2 + \frac{1}{2}(z-\mu_{ab'})^2 \Big) \\
& \geq \kappa_0\frac{C^2_{ak} D_{bl} D_{b'l}}{W^2} \Big( \frac{1}{2}(\mu_{ab}-z) + \frac{1}{2}(z-\mu_{ab'})\Big)^2 \\
& = \kappa_0\frac{C^2_{ak} D_{bl} D_{b'l}}{4W^2} (\mu_{ab}-\mu_{ab'})^2.
\end{align*}
Thus
\begin{multline*}
  [\mC^T \vone]_k [\mD^T \vone]_l
  f\bigg(\frac{[\mC^T \mM_0 \mD]_{kl}}{[\mC^T \vone]_k [\mD^T \vone]_l}\bigg)
  -
  [\mC^T \mN \mD]_{kl} \\
  \leq
  -\frac{\kappa_0}{2} (\mu_{ab} - \mu_{ab'})^2 \frac{C_{ak}^2 D_{bl} D_{b'l}}{W}
  \leq
  -\frac{\kappa_0 \eta^2 \delta}{4 K^2} (\mu_{ab} - \mu_{ab'})^2.
\end{multline*}

This inequality only holds for one particular choice of $k$ and $l$; for other
choices, the left hand side is nonpositive by Jensen's inequality.
Defining
\[
\kappa_1 = \frac{\kappa_0}{4} \min_{a, b \neq b'} (\mu_{ab} - \mu_{ab'})^2,
\]
it follows that
\[
   G(\mC, \mD)
     - \sum_{a,b} p_a \, q_b \, f(\rho^{-1} \mu_{ab})
     \leq - \frac{\kappa_1 \, \eta^2 \delta}{K^2}.
\]
Similarly, if $\mC \notin \sP_\delta$, then the right hand side is
bounded by
\[
  - \kappa_2 \, \eta^2\delta / L^2
\]
where
\[
\kappa_2 = \frac{\kappa_0}{4} \min_{a \neq a', b} (\mu_{ab} - \mu_{a'b})^2.
\]
The result of the proposition follows with
$\kappa = \min(\kappa_1,\kappa_2)/\max\{K^2, L^2\}$.
\end{proof}

\subsubsection{Consistency}

We are now ready to state and prove the formal consistency theorem.

\begin{proof}[Proof of Theorem~\ref{T:consistency}]
Fix $\delta > 0$ and define $\sP_\delta$ and $\sQ_\delta$ as
in~\eqref{E:confusion-nbhd}.
We will show that if $(\vg, \vh) \in \sJ_\varepsilon$ and
if $(\mC(\vg), \mD(\vh)) \notin (\sP_\delta, \sQ_\delta)$,
then $F(\vg, \vh) < F(\vc, \vd)$ with probability approaching one.  Moreover,
this inequality holds uniformly over all such choices of $(\vg, \vh)$.  Since
$\delta$ is arbitrary, this implies that $\mC(\vhg)$ and $\mD(\vhh)$ converge
to permutations of diagonal matrices, i.e.~the proportions of misclassified
rows and columns converge to zero.

Set $r_n = \sup_{\sJ_\varepsilon} |F(\vg, \vh) - G(\mC(\vg),
\mD(\vh)|$.  Suppose $(\vg, \vh) \in \sJ_\varepsilon$.  In this case,
\begin{align*}
  F(\vg, \vh) - F(\vc, \vd)
    & \leq 2 r_n
       + \{G(\mC(\vg), \mD(\vh)) - G(\mC(\vc), \mD(\vd))\} \\
    &= 2 r_n
       + \big\{G(\mC(\vg), \mD(\vh))
     - \sum_{a,b} [\mC \vone]_a [\mD \vone]_b f([\mM_0]_{ab}) \big\}.
\end{align*}

Pick $\eta > 0$ smaller than $\min_a \{ p_a \}$ and $\min_b\{ q_b \}$.
By assumption, the true row and
column class proportions converge to $\vp$ and
$\vq$.  Thus,
for all $\vg \in K^{m}$ and $\vh \in L^n$, for $n$ large enough,
$[\mC(\vg)]_a \geq \eta$ and $[\mD(\vh)]_b \geq \eta$; this holds uniformly
over all choices of $(\vg, \vh)$.

Applying
Proposition~\ref{L:pop-perturb}, to the second term in the inequality, we get that
with probability approaching one,
\[
  F(\vg, \vh) - F(\vc, \vd) \leq 2 r_n - \kappa \eta^2 \delta
\]
for all $(\vg, \vh) \in \sJ_\varepsilon$ such that
$(\mC(\vg), \mD(\vh)) \notin \sP_\delta \times \sQ_\delta$.
By Proposition~\ref{L:pop-crit}, $r_n \toP 0$.  Thus, with probability
approaching one, $(\mC(\vhg), \mD(\vhh)) \in \sP_\delta \times \sQ_\delta$.  Since this
result holds for all $\delta$, all limit points of $\mC(\vhg)$ and $\mD(\vhh)$
must be permutations of diagonal matrices.
\end{proof}

\subsubsection{Empirical treatment of \texorpdfstring{$\rho$}{rho}}\label{S:rho}

For the Poisson and Gaussian relative entropy
functions~\eqref{E:f-poisson}~and~\eqref{E:f-gaussian}, the maximizer of the
criterion function~\eqref{E:criterion-general} does not depend on the scale
factor $\rho$.  This is immediately obvious in the Gaussian case.  For the
Poisson case, the function
$f_{\text{Poisson}}(\mu/\rho)=\frac{1}{\rho}\mu\log(\mu) -
\frac{1}{\rho}\mu(1+\log(\rho))$.  When summed over all biclusters, the
second term in this sum is equal to a constant so the value of $\mu$ which
maximizes $f_{\text{Poisson}}(\mu/\rho)$ does not depend on the value of
$\rho$.

This is not the case for the Binomial relative entropy function~\eqref{E:f-bernoulli}, but
the parameter $\rho$ is not identifiable in practice so it does not make sense
to try to estimate it.  For our simulations we use which maximizes
$\rho = 1$ in the fitting procedure, regardless of the true scale factor
for the mean matrix $\mM$.  Even though in the simulations
the identifiability condition doesn't hold for this choice of $\rho$,
we still get consistency, because the maximizer of the criterion with
$f_\text{Bernoulli}(\mu)$ is close to the maximizer with
$f_\text{Poisson}(\mu)$.  See \citet{perry12} for discussion of
a related phenomenon.

\subsection{Additional Technical Results}\label{S:technical}
\begin{lemma}\label{L:refined-bernstein}

For each $n$, let $X_{n,m}$, $1 \leq m \leq n$, be independent random
variables with $\E X_{n,m} = 0$.  Let $\rho_n$ be a sequence of
positive numbers.  Let $\sI_n$ be a subset of
the powerset $2^{[n]}$, with
\(
  \inf \{ |I| : I \in \sI_n\} \geq L_n.
\)
Suppose
\begin{enumerate}[(i)]
\item \label{a:finite-var}
\(
    \frac{1}{n \rho_n}
    \sum_{m=1}^{n}
      \E\abs{X_{n,m}}^2
\)
is uniformly bounded in $n$;

\item \label{a:lindeberg} For all $\varepsilon > 0$,
\(
    \frac{1}{n \rho_n^2}
    \sum_{m=1}^{n}
      \E(\abs{X_{n,m}}^2 ; \abs{X_{n,m}} > \varepsilon \sqrt{n} \rho_n)
  \to 0;
\)

\item \label{a:bigsize}
\(
  \varlimsup_{n \to \infty} \frac{n}{L_n} < \infty;
\)

\item \label{a:smallset}
\(
  \varlimsup_{n \to \infty} \frac{\log \abs{\sI_n}}{\sqrt{n}} < \infty.
\)

\item \label{a:sparseness}
\(
  \varlimsup_{n \to \infty} \rho_n \sqrt{n} = \infty.
\)

\end{enumerate}
Then
\[
\sup_{I \in \sI_n} \Big| \frac{1}{\rho_n \abs{I}} \sum_{m \in I} X_{n,m} \Big| \toP 0.
\]
\end{lemma}
\begin{proof}
Let $\varepsilon > 0$ be arbitrary.
Define
\(
  Y_{n,m} = \rho_n^{-1} X_{n,m} \I(\abs{X_{n,m}} \leq \varepsilon \sqrt{n} \rho_n),
\)
and note that
\begin{multline*}
  \Pr( Y_{n,m} \neq \rho_n^{-1} X_{n,m} \text{ for some $1 \leq m \leq n$})
    \leq \sum_{m=1}^{n} \Pr(\abs{X_{n,m}} > \varepsilon  \sqrt{n} \rho_n) \\
    \leq \frac{1}{\varepsilon^2 n \rho_n^2} \sum_{m=1}^{n} \E(\abs{X_{n,m}}^2 ; \abs{X_{n,m}} > \varepsilon \sqrt{n} \rho_n).
\end{multline*}

Fix an arbitrary $t > 0$.
Set
\(
  \mu_{n,m} = \E Y_{n,m}
\)
and for $I \in \sI_n$ define
\[
  \mu_{n}(I) = \frac{1}{\abs{I}} \sum_{m \in I} \mu_{n,m}.
\]
For $I \in \sI_n$, write
\[
  \Pr\Big(\sum_{m \in I} Y_{n,m} > t \, \abs{I}\Big)
  =
  \Pr\Big(\sum_{m \in I} (Y_{n,m} - \mu_{n,m})
    > \abs{I} \big(t - \mu_n(I)\big)\Big).
\]
Note that since $\E X_{n,m} = 0$, it follows that
\begin{align*}
  \abs{\mu_{n,m}}
    & =
    \abs{-\E(\rho_n^{-1} X_{n,m} ; \abs{X_{n,m}} > \varepsilon \sqrt{n} \rho_n)}\\
   & \leq
    \frac{1}{\varepsilon \sqrt{n} \rho^2_n}
    \E(\abs{X_{n,m}}^2 ; \abs{X_{n,m}} > \varepsilon \sqrt{n} \rho_n).
\end{align*}
Thus, by~(\ref{a:lindeberg}) and (\ref{a:bigsize}) we have that
\(
  \sup_{I \in \sI_n} \{ \abs{
  \mu_n(I)} \} \to 0;
\)
in particular, for $n$ large enough,
\(
  \sup_{I \in \sI_n} \{ \abs{\mu_n(I)} \} < \frac{t}{2}.
\)
Consequently, for $n$ large enough, uniformly for all $I$,
\[
  \Pr\Big(\sum_{m \in I} Y_{n,m} > t \, \abs{I}\Big)
  \leq
  \Pr\Big(\sum_{m \in I} (Y_{n,m} - \mu_{n,m})
    > t \, \abs{I} / 2 \Big).
\]
Similarly,
\[
  \Pr\Big(\sum_{m \in I} Y_{n,m} < - t \, \abs{I}\Big)
  \leq
  \Pr\Big(\sum_{m \in I} (Y_{n,m} - \mu_{n,m})
    < - t \, \abs{I} / 2 \Big).
\]

We apply Bernstein's inequality to the bound.  Define
\(
  \sigma_{n,m}^2 = \E(Y_{n,m} - \mu_{n,m})^2
\)
and $v_n(I) = \sum_{m \in I} \sigma_{n,m}^2$.
Note that $\abs{Y_{n,m} - \mu_{n,m}} \leq 2 \varepsilon \sqrt{n}$.
By Bernstein's inequality,
\[
  \Pr\Big( \Big| \sum_{m \in I} (Y_{n,m} - \mu_{n,m}) \Big| > t \, \abs{I} / 2 \Big)
    \leq
    2 \exp\Big\{
      -
      \frac{t^2 \abs{I}^2 / 8}{v_n(I) + \varepsilon t \abs{I} \sqrt n / 3}
    \Big\}.
\]
By~(\ref{a:finite-var}), (\ref{a:smallset}), and (\ref{a:sparseness}), it follows that for $n$ large enough,
$v_n(I) < \varepsilon t \abs{I} \sqrt{n} / 3,$ so
\[
  \Pr\Big( \Big| \sum_{m \in I_n} (Y_{n,m} - \mu_{n,m}) \Big| > t \abs{I} / 2 \Big)
    \leq
    2\exp\Big\{
      -
      \frac{t \abs{I}}{8 \varepsilon \sqrt n}
    \Big\}.
\]
We use this bound for each $I$ to get the union bound:
\begin{align*}
  \Pr\Big(\sup_{I \in \sI_n} \Big| \frac{1}{\abs{I}} \sum_{m \in I}
Y_{n,m} \Big| > t\Big)
  & \leq
    2\abs{\sI_n}
    \exp\Big\{
      -
      \frac{t L_n}{8 \varepsilon \sqrt n}
    \Big\} \\
  & =
     2\exp\Big\{
     \log \abs{\sI_n}
      -
      \frac{t L_n}{8 \varepsilon \sqrt n}
    \Big\}.
\end{align*}
By~(\ref{a:bigsize}) and (\ref{a:smallset}), it is possible to choose $\varepsilon$ such
that the right hand side goes to zero.  It follows then that
\begin{align*}
\Pr\Big(\sup_{I \in \sI_n} \Big| \frac{1}{\rho_n\abs{I}} \sum_{m \in I}
X_{n,m} \Big| > t\Big) &
\leq
\Pr( Y_{n,m} \neq \rho_n^{-1} X_{n,m} \text{ for some $1 \leq m \leq n$}) \\
& +
\Pr\Big(\sup_{I \in \sI_n} \Big| \frac{1}{\abs{I}} \sum_{m \in I}
Y_{n,m} \Big| > t\Big)\\
& \to 0.
\end{align*}

\end{proof}

\begin{lemma}[Refined Jensen's Inequality]\label{L:refined-jensen}
  Let $f : \R \to \R$ be twice differentiable and
  let $\sN$ be a convex set in $\R$. If $x_1, \dotsc, x_n$
  are points in $\sN$, and if $w_{1}, \dotsc, w_{n}$ are nonnegative
  numbers summing to one, then
  \[
\sum_{i=1}^{n} w_i f(x_i)
-
f(z)
\geq
\frac{1}{2}
\inf_{y \in \sN} f''(y)
\sum_{i=1}^{n} w_i (x_i - z)^2,
\]
  where $z = \sum_{i=1}^{n} w_i x_i$.
\end{lemma}

\begin{proof}
Define $\kappa_0 = \inf_{y \in \mathcal{N}} f''(y)$ and use the bound
\[
f(x_i) \geq f(z) + f'(z) (x_i - z) + \frac{\kappa_0}{2} (x_i - z)^2.
\]
\end{proof}

\subsection{Finite Sample Results}\label{S:finite-sample}
In this appendix we derive a finite sample tail bound for the probability that
the class assignments that maximize the profile likelihood are close to the true class labels.
To proceed in this setting, we make stronger distributional assumptions than in the asymptotic case.
Specifically, we assume here that the entries $X_{ij}|\vc, \vd$ follow a Gaussian distribution
with mean $\mu_{c_i d_j}$ and finite variance $\sigma^2$.  We proceed with the notation from the main text.

\begin{proposition}\label{fs:conc-bound}
For all $\varepsilon>0$, if $t < \sigma$ then
\[
\Pr \Big(  \sup_{\sJ_\varepsilon} \| \mR(\vg,\vh) \|_\infty > t \Big)
    \leq 2 K^{m+1}L^{n+1}
        \exp \Big( - \frac{L_n t^2}{4 \sigma^2} \Big),
\]
and if $t \geq \sigma$ then
\[
\Pr \Big(  \sup_{\sJ_\varepsilon} \| \mR(\vg,\vh) \|_\infty > t \Big)
    \leq 2 K^{m+1}L^{n+1}
        \exp \Big( - \frac{L_n t}{4 \sigma} \Big)
\]
where $\|\mA\|_\infty = \max_{k,l} |A_{kl}|$ for any matrix $\mA$.
\end{proposition}
\begin{proof}
If the entries $X_{ij}$ follow a Gaussian distribution with mean $\mu_{c_i d_j}$ and variance $\sigma^2$ then
\[
\E \Big( |X_{ij} - \mu_{c_i d_j}|^l \Big)
    \leq
        \frac{\sigma^2}{2} \sigma^{l-2} l!
\]
so the conditions of Bernstein's inequality hold.  It follows that for any
bicluster $I$, for all $t>0$,
\begin{align*}
\Pr \Big(
    \abs{\sum_{i,j \in I} X_{ij} - \mu_{c_i d_j}}
    > t|I| \Big)
    &   \leq 2 \exp \Big\{ -\frac{  |I|^2t^2 }{2(\sigma^2|I| + \sigma |I| t)} \Big\}\\
    &   \leq 2 \exp \Big\{ -\frac{  L_n t^2 }{4\max\{\sigma^2,\sigma t\}} \Big\}.
\end{align*}
Applying a union bound,
\[
\Pr \Big(  \sup_{\sJ_\varepsilon} \| \mR(\vg,\vh) \|_\infty > t \Big)
\leq
2 K^{m+1}L^{n+1}\exp \Big\{ -\frac{  L_n t^2 }{4\max\{\sigma^2,\sigma t\}} \Big\}.
\]
\end{proof}

Proposition~\ref{fs:conc-bound} is used to establish a finite sample bound on the
difference between $F(\vg,\vh)$ and its population version.

\begin{proposition}\label{fs:pop-bound}
Under conditions~\ref{A:rate-convex} and~\ref{A:deriv-bound}, for any $t>0$,
\[
\Pr \Big(
    \sup_{\sJ_\varepsilon}
    | F(\vg, \vh) - G(\mC, \mD)|
    >t
    \Big)
    \leq
    \Pr \Big(  \sup_{\sJ_\varepsilon} \| \mR(\vg,\vh) \|_\infty > \frac{t}{c} \Big)
\]
where $c = \sup |f'(\mu)|$ for $\mu$ in $\sM$.
\end{proposition}

Proposition~\ref{fs:pop-bound} is a direct consequence of the fact that $f$ is
locally Lipschitz continuous under conditions~\ref{A:rate-convex}
and~\ref{A:deriv-bound}. The details are similar to proof of
Proposition~\ref{L:pop-crit}.

The next step is to show that population version is maximized at the true class labels.

\begin{proposition}\label{fs:max-truelabels}
  Choose $\tau>0$ such that $\min_{a \neq a',b} (\mu_{ab}-\mu_{a'b})^2 \geq
\tau$ and $\min_{a ,b \neq b'} (\mu_{ab}-\mu_{ab'})^2 \geq \tau$. Then for all
$\varepsilon>0$, for $(\vg,\vh) \in \sJ_{\varepsilon}$ and $(\mC, \mD) \notin
\sP_\delta \cap \sQ_\delta$, $G(\mC,\mD)$ is small in the sense that
\[
G(\mC, \mD) - \sum_{a,b} p_a q_b f(\rho^{-1} \mu_{ab})
\leq - \frac{\tau \varepsilon^2 \delta}{4 \max\{K^2,L^2\}}.
\]
\end{proposition}
The proof of Proposition~\ref{fs:max-truelabels} is similar to the proof of Proposition~\ref{L:pop-perturb} except that the bound
on the difference uses $\varepsilon$ in place of the value $\eta$.  Some details follow.

\begin{proof}[Proof of Proposition~\ref{fs:max-truelabels}]
First note that in Proposition~\ref{L:pop-perturb}, we can let $a$ be the index of the
largest element in column $k$ of matrix $\mC$; then, since we are restricted
to the set $\sJ_\varepsilon$, this element must be at least
as large as
\[
C_{ak} \geq \frac{[\mC^T \vone]_k}{K} \geq \frac{\varepsilon}{K}.
\]
Noting that for the Gaussian relative entropy function $f''(\mu) = 1$ for all $\mu \in \sM$, the remainder of the proof is similar to the proof of Proposition~\ref{L:pop-perturb}.
\end{proof}

We establish a finite sample bound by combining these results.

\begin{proof}{Proof of Theorem~\ref{T:finite-sample}}
Fix $\delta > 0$ and define $\sP_\delta$ and $\sQ_\delta$ as in Proposition~\ref{L:pop-perturb}.
Set $r_n = \sup_{\sJ_\varepsilon} |F(\vg, \vh) - G(\mC(\vg),
\mD(\vh)|$. Suppose $(\vg, \vh) \in \sJ_\varepsilon$. In this case,
\begin{align*}
  F(\vg, \vh) - F(\vc, \vd)
    &\leq 2 r_n
       + \{G(\mC(\vg), \mD(\vh)) - G(\mC(\vc), \mD(\vd))\} \\
    &= 2 r_n
       + \big\{G(\mC(\vg), \mD(\vh))
     - \sum_{a,b} [\mC \vone]_a [\mD \vone]_b f([\mM_0]_{ab}) \big\}.
\end{align*}
Applying Proposition~\ref{fs:max-truelabels} to the second term in the inequality, we get that
\[
F(\vg, \vh) - F(\vc, \vd) \leq 2r_n - \frac{\tau \varepsilon^2 \delta}{4\max\{K^2,L^2\}}
\]
for all $(\vg, \vh) \in \sJ_\varepsilon$ such that $(\mC, \mD) \notin
\sP_\delta \cap \sQ_\delta$.  The result
follows by applying Proposition~\ref{fs:pop-bound}.
\end{proof}

\section{Additional Empirical and Application Results}\label{app-empirical}

This appendix reports additional empirical results for Bernoulli, Gaussian,
and Student's~$t$ distributed data as well as an additional application example and stability results for the
proposed biclustering algorithm.

\subsection{Additional Empirical Results}\label{S:extra-sim}

Figures \ref{binSim}-\ref{tSim} present the average bicluster
misclassification rates for each sample size and Tables
\ref{tabBin}-\ref{tabt} report the standard deviations for the Bernoulli,
Gaussian, and $t$ simulations,
respectively.  Since the normalization for the DI-SIM algorithm is only specified for non-negative data,
the algorithm is run on the un-normalized matrix for the Gaussian and
non-standardized Student's $t$ examples.

For the Bernoulli simulation, we simulate from a block model with $K = 2$ row
clusters and $L = 3$ column clusters.
We vary the number of columns, $n$, between 200 to 1400 and we take the number
of rows as $m = \gamma n$ where $\gamma \in \{ 0.5,1,2 \}$.

We set the row and column class
membership probabilities as $\vp = (0.3, 0.7)$ and $\vq = (0.2, 0.3, 0.5)$.
We choose the matrix of block parameters to be
\[ \mM = \frac{b}{\sqrt{n}} \begin{pmatrix}
0.43 & 0.06 & 0.13  \\
0.10 & 0.34 & 0.17 \end{pmatrix}.\]
where the entries were selected to be on the same scale as Bickel and Chen (2009).
We vary $b$ between 5 and 20.  We generate
the data conditional on the row and column classes as
\(
  X_{ij} \mid \vc, \vd \sim \mathrm{Bernoulli}(\mu_{c_i d_j}).
\)
We initialize all the methods with $250$ random starting values.

\begin{figure}[H]
\begin{center}
\includegraphics[scale = 1, trim = 0mm 20mm 0mm 0mm]{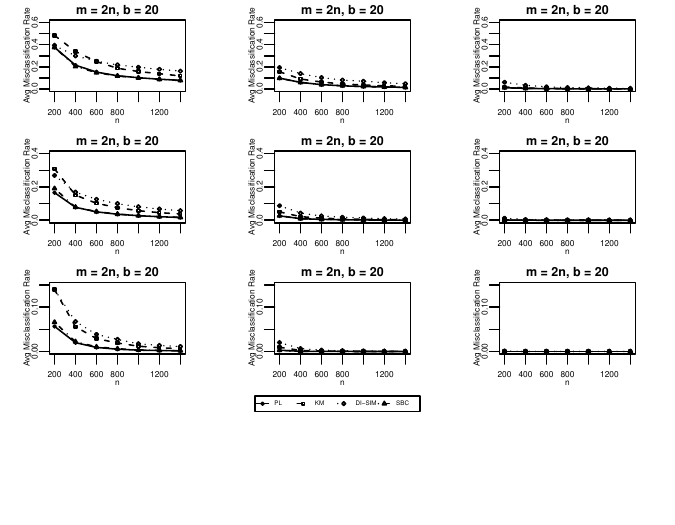}
\end{center}
\caption{Average misclassification rates for Bernoulli example over 100 simulations.}\label{binSim}
\end{figure}

\begin{table}[H]\tiny
\centering
\scalebox{0.9}{
\begin{tabular}{|l|ccc|ccc|ccc|ccc|}
\hline
\multicolumn{13}{|c|}{$m=0.5n$}\\
\hline
		&	\multicolumn{3}{c}{PL}	 &	\multicolumn{3}{|c|}{KM} & \multicolumn{3}{|c|}{DS} & \multicolumn{3}{|c|}{SBC}\\
\hline
n   & b=5   & b=10 & b=20 & b=5   & b=10 & b=20 & b=5   & b=10 & b=20 & b=5   & b=10 & b=20 \\
200	&	0.0834	&	0.0268	&	0.0075	&	0.1141	&	0.0515	&	0.0089	&	0.0760	&	0.0440	&	0.0249	&	0.1193	&	0.0249	&	0.0075	\\
400	&	0.0473	&	0.0145	&	0.0035	&	0.0932	&	0.0168	&	0.0051	&	0.0400	&	0.0346	&	0.0139	&	0.0417	&	0.0125	&	0.0030	\\
600	&	0.0224	&	0.0089	&	0.0021	&	0.0654	&	0.0118	&	0.0027	&	0.0347	&	0.0196	&	0.0075	&	0.0199	&	0.0082	&	0.0022	\\
800	&	0.0168	&	0.0061	&	0.0018	&	0.0266	&	0.0069	&	0.0024	&	0.0273	&	0.0173	&	0.0048	&	0.0135	&	0.0060	&	0.0017	\\
1000	&	0.0118	&	0.0049	&	0.0010	&	0.0156	&	0.0055	&	0.0014	&	0.0218	&	0.0134	&	0.0033	&	0.0098	&	0.0046	&	0.0010	\\
1200	&	0.0094	&	0.0041	&	0.0007	&	0.0122	&	0.0052	&	0.0011	&	0.0205	&	0.0114	&	0.0031	&	0.0086	&	0.0041	&	0.0008	\\
1400	&	0.0079	&	0.0031	&	0.0006	&	0.0100	&	0.0038	&	0.0010	&	0.0161	&	0.0089	&	0.0023	&	0.0076	&	0.0033	&	0.0006	\\
\hline
\multicolumn{13}{c}{}\\
\hline
\multicolumn{13}{|c|}{$m=n$}\\
\hline
n   & b=5   & b=10 & b=20 & b=5   & b=10 & b=20 & b=5   & b=10 & b=20 & b=5   & b=10 & b=20 \\
200	&	0.0528	&	0.0125	&	0.0020	&	0.1084	&	0.0181	&	0.0031	&	0.0514	&	0.0352	&	0.0104	&	0.0833	&	0.0110	&	0.0018	\\
400	&	0.0192	&	0.0053	&	0.0006	&	0.0265	&	0.0080	&	0.0012	&	0.0348	&	0.0147	&	0.0034	&	0.0163	&	0.0050	&	0.0005	\\
600	&	0.0114	&	0.0027	&	0.0003	&	0.0156	&	0.0049	&	0.0007	&	0.0248	&	0.0097	&	0.0018	&	0.0102	&	0.0031	&	0.0003	\\
800	&	0.0071	&	0.0021	&	0.0002	&	0.0107	&	0.0032	&	0.0003	&	0.0179	&	0.0064	&	0.0011	&	0.0070	&	0.0020	&	0.0001	\\
1000	&	0.0052	&	0.0013	&	0.0000	&	0.0073	&	0.0022	&	0.0002	&	0.0140	&	0.0046	&	0.0008	&	0.0052	&	0.0013	&	0.0002	\\
1200	&	0.0043	&	0.0008	&	0.0000	&	0.0067	&	0.0017	&	0.0001	&	0.0120	&	0.0033	&	0.0005	&	0.0040	&	0.0010	&	0.0000	\\
1400	&	0.0034	&	0.0008	&	0.0000	&	0.0057	&	0.0015	&	0.0001	&	0.0097	&	0.0026	&	0.0005	&	0.0032	&	0.0009	&	0.0001	\\
\hline
\multicolumn{13}{c}{}\\
\hline
\multicolumn{13}{|c|}{$m=2n$}\\
\hline
n   & b=5   & b=10 & b=20 & b=5   & b=10 & b=20 & b=5   & b=10 & b=20 & b=5   & b=10 & b=20 \\
200	&	0.0197	&	0.0037	&	0.0000	&	0.0505	&	0.0076	&	0.0000	&	0.0388	&	0.0112	&	0.0025	&	0.0390	&	0.0036	&	0.0000	\\
400	&	0.0062	&	0.0009	&	0.0000	&	0.0120	&	0.0026	&	0.0000	&	0.0212	&	0.0044	&	0.0003	&	0.0074	&	0.0008	&	0.0000	\\
600	&	0.0036	&	0.0004	&	0.0000	&	0.0073	&	0.0016	&	0.0000	&	0.0103	&	0.0025	&	0.0000	&	0.0041	&	0.0004	&	0.0000	\\
800	&	0.0025	&	0.0003	&	0.0000	&	0.0049	&	0.0008	&	0.0000	&	0.0065	&	0.0014	&	0.0001	&	0.0025	&	0.0003	&	0.0000	\\
1000	&	0.0015	&	0.0000	&	0.0000	&	0.0035	&	0.0005	&	0.0000	&	0.0045	&	0.0011	&	0.0001	&	0.0016	&	0.0001	&	0.0000	\\
1200	&	0.0012	&	0.0001	&	0.0000	&	0.0029	&	0.0003	&	0.0000	&	0.0041	&	0.0007	&	0.0000	&	0.0012	&	0.0001	&	0.0000	\\
1400	&	0.0009	&	0.0000	&	0.0000	&	0.0024	&	0.0002	&	0.0000	&	0.0032	&	0.0006	&	0.0000	&	0.0010	&	0.0000	&	0.0000	\\
\hline
\end{tabular}}

\caption{Standard deviations for Bernoulli example over 100 simulations.}\label{tabBin}
\end{table}

For the Gaussian simulation, we simulate from a block model with $K = 2$ row
clusters and $L = 3$ column clusters.
We vary the number of columns, $n$, between 50 to 400 and we take the number
of rows as $m = \gamma n$ where $\gamma \in \{ 0.5,1,2 \}$.

We set the row and column class
membership probabilities as $\vp = (0.3, 0.7)$ and $\vq = (0.2, 0.3, 0.5)$.
We choose the matrix of block parameters to be
\[
\mM = b \begin{pmatrix}
\phantom{-}0.47 & 0.15 & -0.60  \\
-0.26 & 0.82 & \phantom{-}0.80  \end{pmatrix}
\]
where the entries were simulated from a uniform distribution on $[-1,1]$.
We vary $b$ between 0.5 and 2.  We generate
the data conditional on the row and column classes as
\(
  X_{ij} \mid \vc, \vd \sim \mathrm{Gaussian}(\mu_{c_i d_j}, \sigma=1).
\)
We initialize all the methods with $100$ random starting values.

\begin{figure}[ht]
\begin{center}
\includegraphics[scale = 1, trim = 0mm 20mm 0mm 0mm]{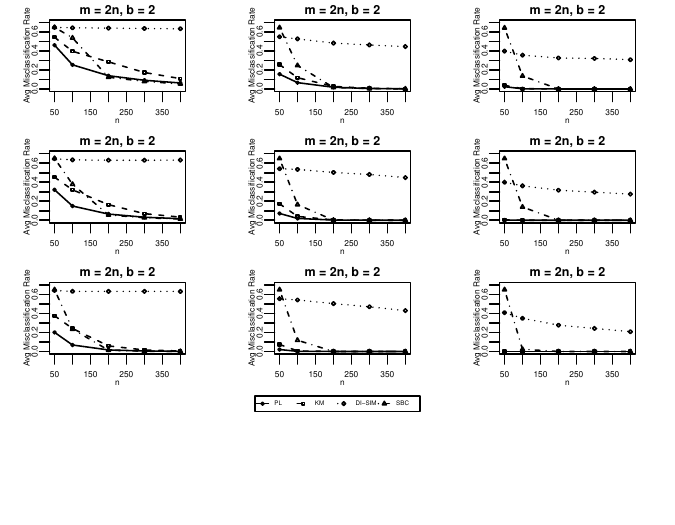}
\end{center}
\caption{Average misclassification rates for Gaussian example over 500 simulations.}\label{normSim}
\end{figure}

\begin{table}[H]\tiny
\centering
\scalebox{0.9}{
\begin{tabular}{|l|ccc|ccc|ccc|ccc|}
\hline
\multicolumn{13}{|c|}{$m=0.5n$}\\
\hline
		&	\multicolumn{3}{c}{PL}	 &	\multicolumn{3}{|c|}{KM} &	 \multicolumn{3}{|c|}{DS} & \multicolumn{3}{|c|}{SBC}\\	
\hline
n   & b=0.5	&	b=1	&	b=2 & b=0.5	&	b=1	&	b=2 & b=0.5	&	b=1	&	b=2 & b=0.5	&	b=1	&	b=2\\
50	&	0.1142	&	0.0718	&	0.0306	&	0.0928	&	0.0923	&	0.0491	&	0.0681	&	0.1105	&	0.1472	&	0.0642	&	0.0732	&	0.0836	\\
100	&	0.0610	&	0.0345	&	0.0051	&	0.0657	&	0.0675	&	0.0074	&	0.0580	&	0.1243	&	0.1374	&	0.1802	&	0.2503	&	0.2487	\\
200	&	0.0329	&	0.0127	&	0.0005	&	0.0722	&	0.0181	&	0.0005	&	0.0526	&	0.1251	&	0.1049	&	0.0289	&	0.0120	&	0.0003	\\
300	&	0.0235	&	0.0052	&	0.0000	&	0.0552	&	0.0068	&	0.0000	&	0.0481	&	0.1195	&	0.1029	&	0.0194	&	0.0043	&	0.0000	\\
400	&	0.0169	&	0.0023	&	0.0000	&	0.0363	&	0.0027	&	0.0000	&	0.0502	&	0.1134	&	0.1200	&	0.0146	&	0.0017	&	0.0000	\\
\hline
\multicolumn{13}{c}{}\\
\hline
\multicolumn{13}{|c|}{$m=n$}\\
\hline
n   & b=0.5	&	b=1	&	b=2 & b=0.5	&	b=1	&	b=2 & b=0.5	&	b=1	&	b=2 & b=0.5	&	b=1	&	b=2 \\
50	&	0.0968	&	0.0503	&	0.0080	&	0.0856	&	0.0975	&	0.0122	&	0.0627	&	0.1171	&	0.1541	&	0.0563	&	0.0596	&	0.0621	\\
100	&	0.0441	&	0.0175	&	0.0008	&	0.0710	&	0.0343	&	0.0010	&	0.0600	&	0.1193	&	0.1365	&	0.2415	&	0.2503	&	0.2546	\\
200	&	0.0235	&	0.0030	&	0.0000	&	0.0629	&	0.0046	&	0.0000	&	0.0447	&	0.1165	&	0.1188	&	0.0196	&	0.0025	&	0.0000	\\
300	&	0.0124	&	0.0007	&	0.0000	&	0.0285	&	0.0011	&	0.0000	&	0.0424	&	0.1166	&	0.1202	&	0.0107	&	0.0005	&	0.0000	\\
400	&	0.0075	&	0.0002	&	0.0000	&	0.0132	&	0.0003	&	0.0000	&	0.0428	&	0.1147	&	0.1261	&	0.0064	&	0.0001	&	0.0000	\\
\hline
\multicolumn{13}{c}{}\\
\hline
\multicolumn{13}{|c|}{$m=2n$}\\
\hline
n   & b=0.5	&	b=1	&	b=2 & b=0.5	&	b=1	&	b=2 & b=0.5	&	b=1	&	b=2 & b=0.5	&	b=1	&	b=2 \\
50	&	0.0783	&	0.0235	&	0.0013	&	0.0795	&	0.0753	&	0.0013	&	0.0596	&	0.1041	&	0.1558	&	0.0519	&	0.0519	&	0.0519	\\
100	&	0.0323	&	0.0040	&	0.0000	&	0.0839	&	0.0086	&	0.0000	&	0.0507	&	0.1084	&	0.1459	&	0.2528	&	0.2358	&	0.0837	\\
200	&	0.0106	&	0.0004	&	0.0000	&	0.0310	&	0.0006	&	0.0000	&	0.0409	&	0.1153	&	0.1329	&	0.0094	&	0.0002	&	0.0000	\\
300	&	0.0045	&	0.0000	&	0.0000	&	0.0085	&	0.0000	&	0.0000	&	0.0353	&	0.1172	&	0.1416	&	0.0035	&	0.0000	&	0.0000	\\
400	&	0.0020	&	0.0000	&	0.0000	&	0.0041	&	0.0000	&	0.0000	&	0.0354	&	0.1177	&	0.1469	&	0.0014	&	0.0000	&	0.0000	\\
\hline
\end{tabular}}

\caption{Standard deviations for Gaussian example over 500 simulations.}\label{tabNorm}
\end{table}

For the non-standardized Student's $t$ simulation, we use the same parameters as in the Gaussian simulation and we generate
the data conditional on the row and column classes as
\(
  X_{ij} \mid \vc, \vd \sim \mathrm{t}(\mu_{c_i d_j}, \sigma=1)
\)
with four degrees of freedom.  We initialize all the methods with $100$ random starting values.

\begin{figure}[ht]
\begin{center}
\includegraphics[scale = 1, trim = 0mm 20mm 0mm 0mm]{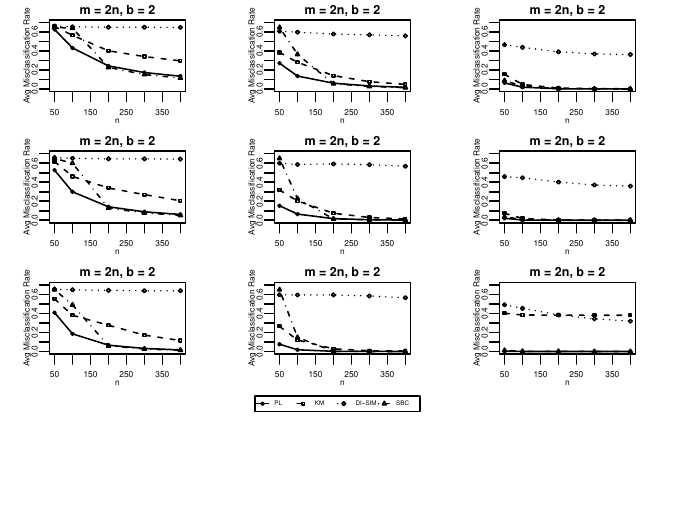}
\end{center}
\caption{Average misclassification rates for $t$ example over 500 simulations.}\label{tSim}
\end{figure}

\begin{table}[H]\tiny
\centering
\scalebox{0.9}{
\begin{tabular}{|l|ccc|ccc|ccc|ccc|}
\hline
\multicolumn{13}{|c|}{$m=0.5n$}\\
\hline
		&	\multicolumn{3}{c}{PL}	 &	\multicolumn{3}{|c|}{KM} &	 \multicolumn{3}{|c|}{DS} & \multicolumn{3}{|c|}{SBC}\\	
\hline
n   & b=0.5   & b=1 & b=2 & b=0.5   & b=1 & b=2 & b=0.5   & b=1 & b=2 & b=0.5   & b=1 & b=2\\
50	&	0.0958	&	0.1018	&	0.0526	&	0.0792	&	0.1028	&	0.1176	&	0.0662	&	0.0845	&	0.1426	&	0.0643	&	0.0669	&	0.1000	\\
100	&	0.0877	&	0.0485	&	0.0175	&	0.0884	&	0.0871	&	0.0775	&	0.0524	&	0.0838	&	0.1394	&	0.0694	&	0.2431	&	0.0535	\\
200	&	0.0429	&	0.0233	&	0.0040	&	0.0620	&	0.0892	&	0.0430	&	0.0409	&	0.0836	&	0.1138	&	0.0441	&	0.0227	&	0.0239	\\
300	&	0.0298	&	0.0137	&	0.0013	&	0.0620	&	0.0831	&	0.0322	&	0.0370	&	0.0880	&	0.1020	&	0.0228	&	0.0116	&	0.0232	\\
400	&	0.0235	&	0.0093	&	0.0007	&	0.0640	&	0.0769	&	0.0192	&	0.0335	&	0.0922	&	0.0910	&	0.0189	&	0.0078	&	0.0126	\\
\hline
\multicolumn{13}{c}{}\\
\hline
\multicolumn{13}{|c|}{$m=n$}\\
\hline
n   & b=0.5   & b=1 & b=2 & b=0.5   & b=1 & b=2 &  b=0.5   & b=1 & b=2 & b=0.5   & b=1 & b=2\\
50	&	0.1073	&	0.0724	&	0.0229	&	0.0798	&	0.0895	&	0.1034	&	0.0561	&	0.0877	&	0.1439	&	0.0561	&	0.0574	&	0.0762	\\
100	&	0.0690	&	0.0344	&	0.0052	&	0.0856	&	0.1021	&	0.0633	&	0.0452	&	0.0808	&	0.1384	&	0.1358	&	0.2477	&	0.0416	\\
200	&	0.0324	&	0.0108	&	0.0009	&	0.0659	&	0.0979	&	0.0298	&	0.0372	&	0.0691	&	0.1206	&	0.0444	&	0.0097	&	0.0130	\\
300	&	0.0220	&	0.0048	&	0.0003	&	0.0668	&	0.0706	&	0.0194	&	0.0298	&	0.0761	&	0.1006	&	0.0176	&	0.0039	&	0.0003	\\
400	&	0.0155	&	0.0024	&	0.0002	&	0.0827	&	0.0556	&	0.0134	&	0.0278	&	0.0846	&	0.1032	&	0.0125	&	0.0019	&	0.0134	\\
\hline
\multicolumn{13}{c}{}\\
\hline
\multicolumn{13}{|c|}{$m=2n$}\\
\hline
n   & b=0.5   & b=1 & b=2 & b=0.5   & b=1 & b=2 & b=0.5   & b=1 & b=2 & b=0.5   & b=1 & b=2\\
50	&	0.0976	&	0.0514	&	0.0086	&	0.0843	&	0.1057	&	0.0659	&	0.0520	&	0.0828	&	0.1343	&	0.0519	&	0.0519	&	0.0484	\\
100	&	0.0528	&	0.0147	&	0.0007	&	0.0774	&	0.1141	&	0.0440	&	0.0443	&	0.0702	&	0.1357	&	0.2231	&	0.2397	&	0.0225	\\
200	&	0.0211	&	0.0035	&	0.0002	&	0.0752	&	0.0807	&	0.0288	&	0.0342	&	0.0646	&	0.1189	&	0.0191	&	0.0025	&	0.0139	\\
300	&	0.0114	&	0.0009	&	0.0000	&	0.0989	&	0.0469	&	0.0216	&	0.0274	&	0.0746	&	0.0965	&	0.0096	&	0.0007	&	0.0000	\\
400	&	0.0068	&	0.0003	&	0.0000	&	0.1086	&	0.0361	&	0.0176	&	0.0260	&	0.0868	&	0.0938	&	0.0058	&	0.0003	&	0.0000	\\
\hline
\end{tabular}}

\caption{Standard deviations for $t$ example over 500 simulations.}\label{tabt}
\end{table}

Similar to the Poisson simulation, biclustering based on the profile log-likelihood criterion performs at least as well as the other methods and shows signs of convergence in all three examples.  These results provide further verification of the theoretical findings and support the use of biclustering based on the profile log-likelihood criterion.

\FloatBarrier

\subsection{Additional Application - MovieLens}\label{S:extra-app}


Since consumer habits likely vary depending on products, biclustering review-website data
can help identify structure in
the data and identify groups of consumers and groups of products with similar
patterns.  As an application of this we apply biclustering to the MovieLens
dataset \citep{movielens}.

The MovieLens dataset consists of 100,000 movie reviews on 1682 movies by 943
users.  Each user has rated at least 20 movies and each movie is rated on a
scale from one to five.  In addition to the review rating, the release date
and genre of each movie is available as well as some demographic information
about each user including gender, age, occupation and zip code.

In order to retain customers, movie-renting services strive to recommend new
movies to individuals who are likely to view them.  Since most users only
review movies that they have already seen, we can use the structure of the
user-movie review matrix to identify associations between users and viewing
habits of movies.  Specifically, we consider the 943$\times$1682 binary matrix
$\mX$ where $X_{ij} = 1$ if user $i$ has rated movie $j$ and $X_{ij} = 0$
otherwise.  To find structure in $\mX$, we biclustered the rows and columns of
$\mX$ using the profile likelihood based on the Bernoulli
criterion~\eqref{E:f-bernoulli}.

To determine a reasonable selection for the number of biclusters we varied the
number of user groups, $K$, and the number of movie groups, $L$, each from
$1$ to $10$.  For each combination of $K$ and $L$, we computed the optimal cluster assignments based on 250 random starting values.  Figure~\ref{screeMovieLens} presents the scree plots as functions of $K$ and $L$.  From the left scree plot, we see little change to the deviance when increasing $L$ beyond $4$.  From the right scree plot, the deviance does not decrease much beyond when increasing $K$ beyond $3$.  Based on these two plots, we set $K = 3$ and $L = 4$.  For $K = 3$ and $L = 4$, we experimented with using up to 2000 random starting values, but found no change to the resulting log-likelihood.  

\begin{figure}[H]
  \includegraphics[scale=0.8]{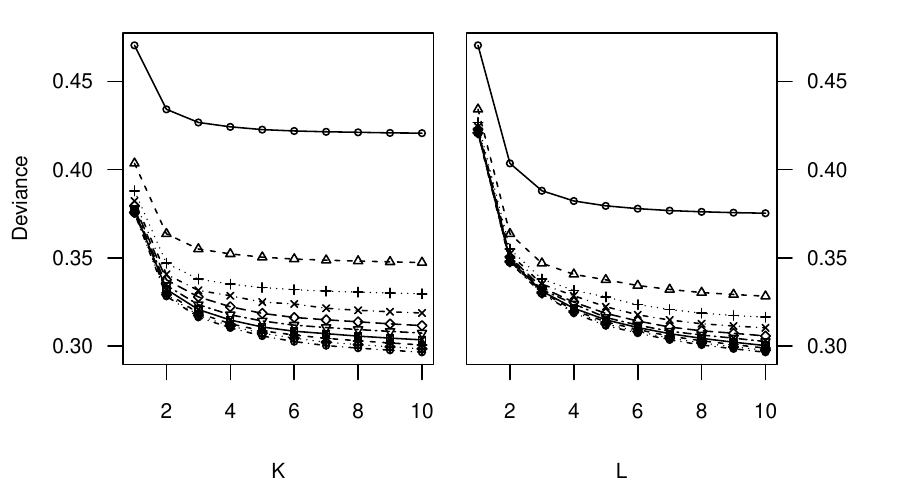}
  \caption{MovieLens likelihood for different values of $K$ and $L$}
  \label{screeMovieLens}
\end{figure}

\begin{figure}[H]
\begin{center}
\includegraphics[trim = 0mm 20mm 0mm 0mm,scale=.9]{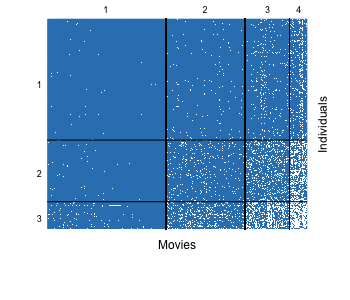}
\end{center}
\caption{Heatmap generated from MovieLens data reflecting the varying review
patterns in the different biclusters. Blue identifies movies with no review and white identifies rated movies.}\label{hmMovieLens}
\end{figure}

We compared the resulting bicluster assignments to those found by DI-SIM, KM and SBC when $K = 3$ and $L = 4$.
The cluster assignments varied between the four methods, with the least amount of disagreement
between PL and SBC (73.0\% of cluster assignments agreed)
and the most disagreement between Di-Sim and the profile-likelihood method (24.2\% of clusters agreed).

Using the built-in package functions, SBC and CBC select different values of $K$ and $L$.
SBC selects $K = 13$ and $L = 14$.  The cluster assignments result in a Rand Index 
of 0.842 and 0.650 when compared to the PL row and column cluster assignments, respectively.
CBC selected 1 row cluster and 1,519 column clusters, where many of the
column clusters only contained one row/column.  CBC is designed for continuous data
and the default settings do not appear to produce meaningful clusters in this example.  

Figure~\ref{hmMovieLens} presents the heatmap of the data based on the
resulting bicluster assignments from the profile-likelihood method, with the ordering 
of the clusters determined
by the total number of a reviews in each cluster.  Roughly speaking, user
group 3 is consistently active across all movie groups with increasing
activity as the popularity of the movie increases.  The reviewing habits of
user group 2 follow a similar pattern but to a lesser extent.  In contrast,
user group 1 is consistently inactive with the only exceptions being movie
group 4.

The median ages within the user group were 33, 30, and 29, and the
percentages of male users within each group were 68.7\%, 72.8\%, and 77.4\%.
These statistics suggest that there is some age and gender effect on the
reviewing habits of the users.

Table~\ref{tableMovieLens}
reports the top ten movies in each group.  The eclectic mix of genres within
each movie group suggests that the rating behavior of users is not explained
by genre alone.

Figure~\ref{bpMovieLens} presents a boxplot comparing
the distributions of the movie release years for each group.  We can see a
clear ordering of the movie groups by median release date.
It appears that the users in all three groups rate movies from all time periods, but 
reviewing behavior varies based on movie popularity.
The biclusters here suggest that individuals in group~3 are more likely to rate under-reviewed
movies, whereas individuals in group~1 primarily rate popular movies.

\begin{table}[H]\tiny
\centering
\scalebox{0.9}{
\begin{tabular}{|c|c|}
\hline
Group 1	&	Group 2		\\
\hline
Mrs. Parker and the Vicious Circle (1994)                         	&	Santa Clause, The (1994)                        	\\
Miserables, Les (1995)                                        	&	Sleeper (1973)                                 	\\
Lawnmower Man 2: Beyond Cyberspace (1996)                         	&	Sword in the Stone, The (1963)                  	\\
Richie Rich (1994)                                               	&	Cook the Thief His Wife \& Her Lover, The (1989)	\\
Candyman: Farewell to the Flesh (1995)                           	&	Somewhere in Time (1980)                        	\\
Ice Storm, The (1997)                                            	&	Mulholland Falls (1996)                        	\\
Funny Face (1957)                                                 	&	Crumb (1994)                                    	\\
Umbrellas of Cherbourg, The (1964)	&	I.Q. (1994)                                    	\\
My Family (1995)                                                  	&	Legends of the Fall (1994)                      	\\
Top Hat (1935)                                                   	&	Alice in Wonderland (1951)                     	\\

\hline
\multicolumn{2}{c}{}\\
\hline
Group 3 & Group 4 \\
\hline
 Beauty and the Beast (1991)                                                	&	Star Wars (1977)              	\\
Batman (1989)                                                              	&	Contact (1997)                	\\
Young Frankenstein (1974)                                                  	&	Fargo (1996)                 	\\
Star Trek IV: The Voyage Home (1986)                                       	&	Return of the Jedi (1983)    	\\
Citizen Kane (1941)                                                        	&	Liar Liar (1997)              	\\
Fifth Element, The (1997)                                                  	&	English Patient, The (1996)   	\\
Gandhi (1982)                                                              	&	Scream (1996)                 	\\
Face/Off (1997)                                                            	&	Toy Story (1995)             	\\
Dr. Strangelove Or (1963)	&	 Air Force One (1997)         	\\
Tin Cup (1996)                                                             	&	 Independence Day (ID4) (1996)	\\

\hline
\end{tabular}}

\caption{The top ten movies in each cluster based on the total number of reviews.}\label{tableMovieLens}
\end{table}

\begin{figure}[H]\label{bpMovieLens}
\begin{center}
\includegraphics[trim = 0mm 20mm 0mm 0mm,scale=.8]{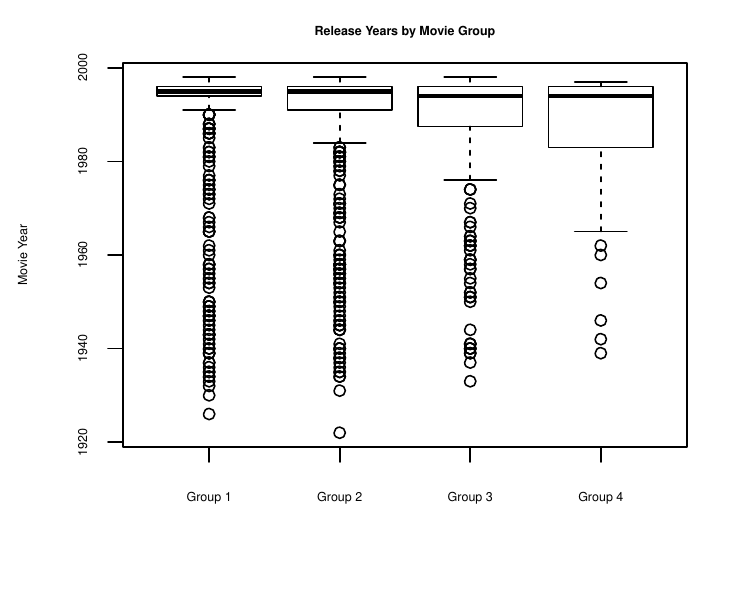}
\end{center}
\caption{Boxplot comparing the different clusters based on movie release dates.}
\end{figure}

\subsection{Algorithm Stability Analysis}\label{S:stability}

For each of the three data applications (GovTrack, AGEMAP, and MovieLens), we
report the local optima found after running our biclustering method with 1000
random initializations. The number of local optima (``modes'') found in each
of the three data sets are $4$, $3$, and $10$, respectively.

For each of the modes, we report the log-likelihood value (``Log-Lik.''), the
number of replicates that found the mode (``Count''), and the row and column
cluster differences (``Row Diff.'' and ``Column Diff.'') when compared to the
best mode out of the 1000 replicates.  In reporting the difference we give
both the absolute difference (``Abs.''), the number of rows or columns with
different labels, and the relative difference (``Rel.''), the proportion of
rows or columns with different labels.

Even in the AGEMAP application, where the best local optimum was rare, the
results of the algorithm are very stable, giving consistent cluster
assignments for over 98\% of the rows and columns. We would not recommend
using our algorithm with just a single random initialization, but it appears
that using 100--1000 random initializations suffices to give reliable results.

\subsubsection*{GovTrack}

The algorithm found the best local optimum 36.8\% of the time. The result of
the algorithm agreed with the best local optimum found in all of the row
labels and at least 98.7\% of the column labels 99.9\% of the time.

\begin{center}
\begin{tabular}{crrrrrr}
\toprule
Mode & Log-Lik. & Count & \multicolumn{2}{c}{Row Diff.}
                        & \multicolumn{2}{c}{Column Diff.} \\
     &          &       & Abs. & Rel. (\%) & Abs. & Rel. (\%) \\
\midrule
1    & $-36246$ & 368   & $-$ & $-$ & $-$  & $-$ \\
2    & $-36248$ & 27    & 0   & 0   &  2 &  $0.4$ \\
3    & $-36430$ & 604   & 0   & 0   &  6 &  $1.3$ \\
4    & $-44732$ & 1     & 0   & 0   & 70 & $15.4$ \\
\bottomrule
\end{tabular}
\end{center}

\noindent
444 rows, 545 columns; $K = 2$ row clusters, $L = 4$ column clusters.

\subsubsection*{AGEMAP}

In all cases, the value of the log-likelihood found was extremely similar,
though the best local optimum was found only 0.2\% of the time.  The row
clusters were always the same, and at least 99.8\% of the column labels agreed
in all reported local optima.

\begin{center}
\begin{tabular}{crrrrrr}
\toprule
Mode & Log-Lik. & Count & \multicolumn{2}{c}{Row Diff.}
                        & \multicolumn{2}{c}{Column Diff.} \\
     &          &       & Abs. & Rel. (\%) & Abs. & Rel. (\%) \\
\midrule
1 & 7408.883 & 2    & $-$  & $-$ & $-$  & $-$ \\
2 & 7408.881 & 909  & 0 & 0 & 18 & $0.1$ \\
3 & 7408.872 & 89   & 0 & 0 & 34 & $0.2$ \\
\bottomrule
\end{tabular}
\end{center}

\noindent
39 rows, 17864 columns; $K = 3$ row clusters, $L = 5$ column clusters.

\clearpage

\subsubsection*{MovieLens}

There was much more variability in the output than in the other two
applications, though the first 5 modes (which appeared in $99.1\%$ of the
replicates) had very similar log-likelihood values, and agreed for over
$98.1$\% of their row labels and over $97.9$\% of their column labels.

\begin{center}
\begin{tabular}{crrrrrr}
\toprule
Mode & Log-Lik. & Count & \multicolumn{2}{c}{Row Diff.}
                        & \multicolumn{2}{c}{Column Diff.} \\
     &          &       & Abs. & Rel. (\%) & Abs. & Rel. (\%) \\
\midrule
1  & $-262910$ & 137 & $-$ &    $-$ & $-$ & $-$ \\
2  & $-262915$ &  54 &  18 &  $1.9$ &   1 &  $0.6$ \\
3  & $-262916$ &   3 &   0 &  $0.0$ &  19 &  $1.1$ \\
4  & $-262919$ &  39 &  17 &  $1.8$ &  13 &  $0.8$ \\
5  & $-262923$ & 758 &   3 &  $0.3$ &  36 &  $2.1$ \\
6  & $-263892$ &   4 & 142 & $15.1$ & 310 & $18.4$ \\
7  & $-264085$ &   1 &  87 &  $9.2$ & 294 & $17.5$ \\
8  & $-265183$ &   2 &  56 &  $5.9$ & 417 & $24.8$ \\
9  & $-265183$ &   1 &  57 &  $6.0$ & 416 & $24.7$ \\
10 & $-265185$ &   1 &  62 &  $6.6$ & 418 & $24.9$ \\
\bottomrule
\end{tabular}
\end{center}

\noindent
943 rows, 1682 columns; $K = 3$ row clusters, $L = 4$ column clusters.


\bibliography{refs5}
\bibliographystyle{apalike}

\end{document}